\documentclass[11pt,a4paper]{article}

\usepackage[a4paper,margin=1in]{geometry}
\usepackage{setspace}
\usepackage{amsmath,amssymb,amsthm,bm,mathtools}
\usepackage{mathrsfs}

\usepackage{graphicx}
\usepackage{booktabs}

\usepackage{algorithm}
\usepackage{algpseudocode}

\usepackage[colorlinks=true]{hyperref}

\usepackage[numbers]{natbib}
\usepackage{authblk}
\theoremstyle{plain}
\newtheorem{theorem}{Theorem}[section]
\newtheorem{proposition}{Proposition}[section]
\newtheorem{corollary}{Corollary}[section]
\theoremstyle{definition}
\newtheorem{definition}{Definition}[section]
\theoremstyle{remark}
\newtheorem{remark}{Remark}[section]

\newcommand{\E}{\mathbb{E}}
\newcommand{\Pp}{\mathbb{P}}
\newcommand{\R}{\mathbb{R}}
\newcommand{\1}{\mathbf{1}}

\newcommand{\norm}[1]{\left\lVert #1 \right\rVert}

\title{ Bayesian Linear Programming under Learned Uncertainty: Posterior Feasibility Guarantees, Scenario Certification, and Applications}

\author[1,2]{Debashis Chatterjee\thanks{Corresponding author: \texttt{debashis.chatterjee@visva-bharati.ac.in}}}

\affil[1]{Department of Statistics, Visva-Bharati University, Santiniketan, India}
\affil[2]{S.\ N.\ Bose National Centre for Basic Sciences, Kolkata, India}

\begin{document}
\maketitle

\begin{abstract}
Linear programming is widely used for decision-making in science, engineering, and operations research, yet in many modern applications the coefficients entering the constraints and objective are not known exactly and must be learned from data. Classical stochastic and robust optimization offer two influential paradigms for handling such uncertainty, but they typically treat the underlying uncertainty description as given and do not directly integrate priors and updated to posteriors guarantees. This paper develops a Bayesian framework for linear programming in which uncertain quantities are modeled probabilistically, updated through observed data, and propagated into optimization through posterior feasibility requirements. We present two complementary computational strategies: a credible-region robustification that converts posterior uncertainty into deterministic protection, and a posterior-scenario approach that uses sampled posterior realizations to construct tractable optimization problems with finite-sample interpretability. We also propose a Monte Carlo certification procedure that provides conservative, data-conditioned assessments of residual infeasibility. Simulation experiments show that the proposed framework substantially improves safety relative to naive plug-in decisions, while a real-data study on single-cell transcriptomic data demonstrates that the approach can produce scientifically interpretable decisions together with explicit uncertainty-aware feasibility diagnostics. The proposed methodology offers a unified bridge between Bayesian learning, optimization under uncertainty, and practical decision certification.
\end{abstract}


\section{Introduction}
Linear programming (LP) is a foundational tool in operations research and engineering, but its classical formulation presumes that the input data (objective and constraint coefficients) are known exactly. In many applications---capacity planning, revenue management, inventory control, energy dispatch, and portfolio allocation---the coefficients are estimated from historical data or produced by predictive models, hence uncertainty is intrinsic.

Two major traditions address uncertainty in optimization. \emph{Stochastic programming} and \emph{chance-constrained programming} treat coefficients as random and demand constraints hold with high probability \citep{prekopa1995,shapiroDentchevaRuszczynski2009}. However, these models often assume the probability law is known. \emph{Robust optimization} instead enforces feasibility for all parameter realizations in a prespecified uncertainty set, providing worst-case protection but sometimes at the cost of conservatism \citep{bentalElGhaouiNemirovski2009,bertsimasSim2004,bertsimasBrownCaramanis2011}.

In this paper we focus on the regime where uncertainty is learned from data and is naturally represented by a \emph{posterior distribution}. Bayesian decision theory provides a coherent framework for conditioning on observed data and choosing actions that optimize posterior expected utility \citep{berger1985}. Yet, for LP feasibility, the translation from posterior uncertainty to \emph{operationally meaningful} guarantees is not standardized.

\paragraph{Research objective.}
The aim of this paper is to develop a statistically principled framework for linear programming when key optimization inputs are not fixed, but are instead learned from data and remain uncertain after estimation. Our goal is to move beyond plug-in optimization by constructing decisions that are not only optimal with respect to learned information, but also accompanied by explicit, data-conditioned feasibility guarantees.

\paragraph{Novelty and contribution.}
The main contribution of this work is to unify Bayesian learning and linear optimization through the notion of posterior feasibility, thereby embedding uncertainty quantification directly into the optimization stage rather than treating it as a separate post hoc correction. In particular, the paper contributes in four directions. First, it formulates a Bayesian notion of feasibility that is conditioned on observed data and therefore naturally aligned with statistical learning. Second, it develops two tractable methodologies for enforcing such feasibility, namely a credible-region robustification strategy and a posterior-scenario strategy. Third, it supplements optimization with a practical certification step that yields conservative empirical assessments of residual infeasibility under the learned posterior. Finally, it demonstrates, through both simulation and real genetic data, that the framework can deliver decisions that are interpretable, uncertainty-aware, and safer than naive plug-in alternatives. Taken together, these contributions position Bayesian linear programming as a coherent meeting point of modern statistical inference and optimization under uncertainty.

\section{Related Work}

Optimization under uncertainty has been studied extensively across several research traditions, including stochastic programming, chance-constrained optimization, robust optimization, and scenario-based optimization. The framework proposed in this paper draws conceptual inspiration from these areas while introducing a distinct Bayesian perspective in which uncertainty is learned from data and propagated into decision guarantees through posterior distributions.

\paragraph{Chance-constrained optimization.}
Chance-constrained programming is one of the earliest approaches to handling uncertainty in optimization problems. In this paradigm, uncertain constraints must hold with high probability under a specified probability distribution. Foundational treatments appear in the work of \citet{prekopa1995} and in modern stochastic programming texts such as \citet{shapiroDentchevaRuszczynski2009}. While these approaches provide elegant probabilistic formulations, they typically assume that the underlying distribution governing uncertainty is known or externally specified. In contrast, our work focuses on situations where the distribution must be learned from data and is therefore naturally represented by a posterior distribution.

\paragraph{Robust optimization.}
Another influential framework is robust optimization, in which feasibility must hold for all parameter realizations within a specified uncertainty set. Seminal contributions include the robust optimization framework of \citet{bentalElGhaouiNemirovski2009} and tractable robust counterparts developed by \citet{bertsimasSim2004}. Subsequent work further explored data-driven and adaptive uncertainty sets \citep{bertsimasBrownCaramanis2011}. Robust optimization provides strong worst-case guarantees but can be overly conservative when uncertainty sets are large or loosely specified. The credible-region robustification approach developed in this paper can be interpreted as a data-driven robust optimization model in which the uncertainty set is derived from a Bayesian posterior credible region.

\paragraph{Scenario-based optimization.}
Scenario approaches provide an alternative strategy in which uncertain constraints are replaced by a finite set of sampled scenarios. The resulting optimization problem enforces feasibility on the sampled realizations, while theoretical results bound the probability of violation. Key theoretical developments were established by \citet{calafioreCampi2006} and \citet{campiGaratti2008}, who derived explicit finite-sample guarantees for convex optimization problems with sampled constraints. Our posterior-scenario formulation directly builds on this theory, with the key distinction that the sampling distribution is the Bayesian posterior conditioned on observed data.

\paragraph{Convex approximations of chance constraints.}
Several authors have studied convex conservative approximations for probabilistic constraints, particularly in the context of convex optimization. Important contributions include the work of \citet{nemirovskiShapiro2006}, who developed tractable approximations of chance constraints under distributional assumptions. These approaches provide computationally convenient surrogates for probabilistic feasibility but generally operate within a frequentist or distributionally specified framework rather than a Bayesian learning paradigm.

\paragraph{Bayesian decision theory and optimization.}
The conceptual foundation for incorporating posterior uncertainty into decision-making lies in Bayesian decision theory \citep{berger1985}. In this framework, optimal decisions are chosen based on posterior distributions conditioned on observed data. While Bayesian decision theory has been widely applied in statistics and machine learning, its integration with classical linear programming formulations has received comparatively little attention. The present work contributes to this direction by introducing posterior feasibility constraints and by combining Bayesian inference with scenario-based feasibility certification.

\paragraph{Contribution relative to existing literature.}
Existing approaches typically treat uncertainty either as known probabilistic structure (stochastic programming), worst-case perturbations (robust optimization), or sampled scenarios from a fixed distribution (scenario optimization). In contrast, the framework developed in this paper integrates \emph{statistical learning} and \emph{optimization under uncertainty} in a unified pipeline. Uncertainty is first learned from data through Bayesian inference, then incorporated into the optimization problem through posterior-feasibility constraints, and finally certified through Monte Carlo diagnostics. This integration provides a principled mechanism for translating statistical uncertainty into operational decision guarantees.
\section{Problem setup}
Let $x\in\R^n$ be the decision vector. Consider an LP with uncertain coefficients indexed by a parameter $\theta\in\Theta$:
\begin{equation}
\label{eq:uncertainLP}
\max_{x\in\mathcal{X}} \; c(\theta)^\top x
\quad\text{s.t.}\quad
A(\theta)x \le b(\theta),
\end{equation}
where $\mathcal{X}\subseteq \R^n$ denotes additional deterministic constraints (e.g., bounds, integrality relaxations, balance constraints). Write the $i$th uncertain inequality as
\begin{equation}
\label{eq:gi}
g_i(x,\theta):= a_i(\theta)^\top x - b_i(\theta) \le 0,\qquad i=1,\dots,m,
\end{equation}
with $a_i(\theta)^\top$ the $i$th row of $A(\theta)$.

\subsection{Bayesian learning of coefficients}
Suppose we observe data $D$ informative about $\theta$ (e.g., historical samples of demands, capacities, yields, regression features and outcomes). A Bayesian model specifies a prior $p(\theta)$ and likelihood $p(D\mid\theta)$, yielding posterior
\[
p(\theta\mid D) \propto p(D\mid\theta)p(\theta).
\]
We assume that posterior sampling is feasible (exactly or approximately), i.e., we can generate i.i.d.\ $\theta^{(1)},\dots,\theta^{(N)}\sim p(\theta\mid D)$.

\subsection{Decision criterion}
A Bayesian decision-maker could maximize posterior expected objective $\E[c(\theta)^\top x\mid D]$ subject to a posterior feasibility constraint. This paper concentrates on feasibility guarantees (which typically dominate safety/operational requirements), while allowing either risk-neutral or risk-averse objectives.

\subsection{Posterior feasibility}
\begin{definition}[Posterior feasibility and violation]
\label{def:postfeas}
Given data $D$, define the \emph{posterior violation probability} of a decision $x$ as
\begin{equation}
\label{eq:violation}
V_D(x) := \Pp_{\theta\mid D}\Big( \ \exists \  i\in\{1,\dots,m\}: g_i(x,\theta) > 0\Big).
\end{equation}
We say that $x$ is \emph{$(1-\alpha)$ posterior-feasible} if $V_D(x)\le \alpha$.
\end{definition}

Thus, posterior feasibility is a chance constraint with probability law equal to the posterior distribution conditioned on the observed data. It is conceptually aligned with Bayesian decision theory \citep{berger1985} and connects to chance-constrained programming \citep{prekopa1995,shapiroDentchevaRuszczynski2009}.

\begin{remark}[Interpretation]
Posterior feasibility is a \emph{conditional} statement: it bounds constraint-violation probability under $p(\theta\mid D)$. This differs from a frequentist guarantee that must hold uniformly over repeated samples of $D$. Both notions can be valuable; the present work focuses on the Bayesian (conditional) one, with explicit accounting of Monte Carlo approximation error.
\end{remark}

\section{Methodology}
\label{sec:methodology}

\subsection{Bayesian chance-constrained LP}
A canonical Bayesian LP with posterior feasibility is:
\begin{equation}
\label{eq:bayesCC}
\max_{x\in\mathcal{X}} \; \E[c(\theta)^\top x \mid D]
\quad\text{s.t.}\quad
V_D(x)\le \alpha.
\end{equation}
Problem \eqref{eq:bayesCC} is generally hard because $V_D(x)$ involves a probability of a union of inequalities.

We now develop two tractable strategies that provide \emph{posterior feasibility guarantees}.

\subsection{Credible-set robustification}
Let $\mathcal{C}_{1-\alpha}(D)\subseteq \Theta$ be a measurable \emph{credible region} such that
\begin{equation}
\label{eq:credmass}
\Pp(\theta\in \mathcal{C}_{1-\alpha}(D)\mid D)\ge 1-\alpha.
\end{equation}
Consider the robustified constraints
\begin{equation}
\label{eq:robustCred}
g_i(x,\theta)\le 0 \quad \forall \theta\in \mathcal{C}_{1-\alpha}(D),\; i=1,\dots,m.
\end{equation}

\begin{proposition}[Credible-set robustification implies posterior feasibility]
\label{prop:credImplies}
If $x$ satisfies \eqref{eq:robustCred} and $\mathcal{C}_{1-\alpha}(D)$ satisfies \eqref{eq:credmass}, then $x$ is $(1-\alpha)$ posterior-feasible, i.e.\ $V_D(x)\le \alpha$.
\end{proposition}

\begin{proof}
If $\theta\in \mathcal{C}_{1-\alpha}(D)$ then all constraints hold by \eqref{eq:robustCred}. Hence violation can occur only when $\theta\notin \mathcal{C}_{1-\alpha}(D)$:
\[
V_D(x)\le \Pp(\theta\notin \mathcal{C}_{1-\alpha}(D)\mid D)\le \alpha,
\]
where the last inequality follows from \eqref{eq:credmass}.
\end{proof}

\paragraph{Gaussian posterior leading to a second-order cone (SOC) form.}
A common case is an approximate Gaussian posterior for each constraint row. Let
\[
u_i(\theta):=\begin{bmatrix}a_i(\theta)\\ b_i(\theta)\end{bmatrix}\in\R^{n+1},
\qquad
z(x):=\begin{bmatrix}x\\ -1\end{bmatrix},
\]
so that
\[
g_i(x,\theta)=u_i(\theta)^\top z(x).
\]
Assume
\[
u_i(\theta)\mid D \approx \mathcal{N}(\bar u_i,\Sigma_i).
\]
To obtain an overall posterior-feasibility guarantee via row-wise credible regions, we use a Bonferroni correction and define
\[
\mathcal{C}_{1-\alpha/m,i}(D):=
\left\{
u:\ (u-\bar u_i)^\top \Sigma_i^{-1}(u-\bar u_i)\le \chi^2_{n+1}(1-\alpha/m)
\right\}.
\]
Requiring $u^\top z(x)\le 0$ for all $u\in \mathcal{C}_{1-\alpha/m,i}(D)$ yields the SOC constraint
\begin{equation}
\label{eq:soc}
\bar u_i^\top z(x) + \kappa_{\alpha/m}\,\norm{\Sigma_i^{1/2}z(x)}_2 \le 0,
\qquad
\kappa_{\alpha/m}:=\sqrt{\chi^2_{n+1}(1-\alpha/m)}.
\end{equation}
By the union bound, enforcing these row-wise robust constraints for all $i=1,\dots,m$ implies an overall posterior-feasibility guarantee at level $1-\alpha$. Thus credible-set robustification can transform Bayesian uncertainty into a deterministic conic program, connecting directly to the robust optimization tradition \citep{bentalElGhaouiNemirovski2009,bertsimasPachamanovaSim2004}. More generally, convex conservative approximations of chance constraints are studied in \citet{nemirovskiShapiro2006}.

\paragraph{Pros/cons.}
Credible-set robustification provides a clean implication (Proposition~\ref{prop:credImplies}), but may be conservative because it enforces constraints for \emph{all} parameters in a region, not merely with high posterior probability.

\subsection{Posterior-scenario approximation with finite-sample bounds}
A second route is to approximate \eqref{eq:bayesCC} by sampled constraints:
\begin{equation}
\label{eq:scenarioLP}
\max_{x\in\mathcal{X}} \; \E[c(\theta)^\top x \mid D]
\quad\text{s.t.}\quad
g_i(x,\theta^{(k)})\le 0,\ \forall i=1,\dots,m,\ \forall k=1,\dots,N,
\end{equation}
where $\theta^{(1)},\dots,\theta^{(N)}\overset{\text{i.i.d.}}{\sim}p(\theta\mid D)$. This is an LP whenever the sampled constraints are linear in $x$.

The scenario approach studies such random constraint sampling and provides explicit bounds on the violation probability of the obtained solution \citep{calafioreCampi2006,campiGaratti2008}. We state a standard scenario-theoretic bound specialized to our posterior setting; this is a conditional-on-data instantiation of the scenario results of \citet{calafioreCampi2006} and \citet{campiGaratti2008}.

\begin{theorem}[Posterior-scenario violation bound (scenario approach)]
\label{thm:scenario}
Assume $\mathcal{X}$ is convex and compact and $g_i(\cdot,\theta)$ is convex for each $\theta$. Let $x_N$ be an optimal solution of \eqref{eq:scenarioLP} based on $N$ i.i.d.\ posterior samples, and let $d$ denote the \emph{support rank} (at most $n$ in many LP settings). Then for any $\varepsilon\in(0,1)$,
\begin{equation}
\label{eq:campiGarattiBound}
\Pp\Big(V_D(x_N) > \varepsilon\Big)\ \le\ \sum_{j=0}^{d-1} \binom{N}{j}\varepsilon^{j}(1-\varepsilon)^{N-j},
\end{equation}
where the probability is over the random posterior sampling used to form \eqref{eq:scenarioLP}.
\end{theorem}

\begin{remark}
Theorem~\ref{thm:scenario} is a direct instantiation of scenario-approach feasibility results for uncertain convex programs \citep{calafioreCampi2006,campiGaratti2008}. In our context, the ``uncertainty distribution'' is the posterior $p(\theta\mid D)$, hence the resulting guarantee targets posterior violation probability.
\end{remark}

\begin{corollary}[Choosing $N$ for a target posterior-feasibility level]
\label{cor:Nchoice}
Fix target posterior violation $\varepsilon$ and confidence $\delta$. If $N$ satisfies
\[
\sum_{j=0}^{d-1} \binom{N}{j}\varepsilon^{j}(1-\varepsilon)^{N-j} \le \delta,
\]
then with probability at least $1-\delta$ (over posterior sampling) we have $V_D(x_N)\le \varepsilon$.
\end{corollary}

\paragraph{Pros/cons.}
Posterior-scenario LPs are often easy to solve (they remain LPs) and provide tunable violation guarantees. Their main cost is sampling and solving a larger LP; their main modeling benefit is reduced conservatism relative to robustification.

\subsection{Post-solution Monte Carlo certification}
Even after solving \eqref{eq:scenarioLP} or the robustified program, it is useful to estimate $V_D(x)$ directly.

Given a candidate solution $\hat x$, draw an independent posterior sample $\{\tilde\theta^{(\ell)}\}_{\ell=1}^{M}\sim p(\theta\mid D)$. Let
\[
\hat V_M(\hat x):=\frac{1}{M}\sum_{\ell=1}^M \1\Big\{\exists i:\ g_i(\hat x,\tilde\theta^{(\ell)})>0\Big\}.
\]

A conservative one-sided $(1-\beta)$ upper confidence bound on the posterior violation probability $V_D(\hat x)$ is obtained via binomial inversion (Clopper--Pearson). Equivalently, this yields a lower confidence bound on the posterior feasibility probability $1-V_D(\hat x)$.

\[
V_D(\hat x) \le \mathrm{BetaInv}\big(1-\beta;\ s+1,\ M-s\big),
\quad s:=\sum_{\ell=1}^M \1\Big\{\exists i:\ g_i(\hat x,\tilde\theta^{(\ell)})>0\Big\},
\]
where $\mathrm{BetaInv}$ is the quantile of a Beta distribution. This step is a practical diagnostic: it does not replace Theorem~\ref{thm:scenario} but complements it with an explicit empirical check.

\section{Algorithmic summary}
\begin{algorithm}[h!]
\caption{Bayesian LP with Posterior Feasibility Guarantees}
\label{alg:blp}
\begin{algorithmic}[1]
\Require Data $D$, posterior sampler for $\theta\mid D$, feasible set $\mathcal{X}$, risk level $\alpha$
\State Fit Bayesian model and obtain access to samples $\theta\sim p(\theta\mid D)$
\State Choose one of:
\Statex \quad (A) Credible-set robustification with credible region $\mathcal{C}_{1-\alpha}(D)$
\Statex \quad (B) Posterior-scenario approximation with sample size $N$ for target $(\varepsilon,\delta)$
\If{(A)}
  \State Construct $\mathcal{C}_{1-\alpha}(D)$ with $\Pp(\theta\in\mathcal{C}_{1-\alpha}(D)\mid D)\ge 1-\alpha$
  \State Solve $\max_{x\in\mathcal{X}} \E[c(\theta)^\top x\mid D]$ s.t.\ $g_i(x,\theta)\le 0$ for all $\theta\in\mathcal{C}_{1-\alpha}(D)$
  \State Output $\hat x$
\Else
  \State Draw $\theta^{(1)},\dots,\theta^{(N)}\sim p(\theta\mid D)$ i.i.d.
  \State Solve the scenario LP \eqref{eq:scenarioLP} to obtain $x_N$
  \State Output $\hat x:=x_N$ with violation guarantee from Theorem~\ref{thm:scenario}
\EndIf
\State Certification: draw $\tilde\theta^{(1)},\dots,\tilde\theta^{(M)}\sim p(\theta\mid D)$ and estimate/upper-bound $V_D(\hat x)$
\end{algorithmic}
\end{algorithm}


\section{Simulation study: posterior-feasibility in Bayesian LP}
\label{sec:sim}

\subsection{Design and evaluation protocol}
\label{subsec:sim_design}
We consider a profit-maximizing production LP with $n=18$ decision variables and $m=7$ resource constraints,
\begin{align}
\max_{x\in\mathbb{R}^n}\quad & p^\top x
\qquad\text{s.t.}\qquad
A x \le b,\quad 0\le x \le x_{\max}\mathbf{1},
\label{eq:sim_lp}
\end{align}
where $A\in\mathbb{R}^{m\times n}$ and $p\in\mathbb{R}^n$ have positive entries and are fixed throughout a Monte Carlo run, while the resource capacities $b\in\mathbb{R}^m$ are uncertain and learned from data.

\paragraph{Data-generating process.}
In each trial, we draw a ``true'' linear model for each resource $j\in\{1,\dots,m\}$:
\[
b_j = x_{\text{ctx}}^\top \beta_j^\star + \varepsilon_j,\qquad \varepsilon_j\sim\mathcal{N}(0,(\sigma_j^\star)^2),
\]
with a context vector $x_{\text{ctx}}\in\mathbb{R}^{d}$ (including an intercept; here $d=6$), and we generate $N_{\text{obs}}=90$ historical observations to learn $\beta_j^\star$ and $(\sigma_j^\star)^2$. We then solve \eqref{eq:sim_lp} using a capacity vector $\hat b$ produced by each method below.

\paragraph{Bayesian learning model.}
For each resource $j$, we fit a conjugate Normal--Inverse-Gamma Bayesian linear regression, yielding an explicit posterior predictive distribution for the next-period capacity $b_j \mid D$ (Student-$t$). Posterior samples are used for scenario-based optimization and for an independent Monte Carlo ``certificate'' of posterior violation.

\paragraph{Methods compared.}
We evaluate five approaches:
\begin{itemize}
\item \textbf{PM (Posterior Mean plug-in):} $\hat b_j=\mathbb{E}[b_j\mid D]$ (risk-ignorant baseline).
\item \textbf{CR (Credible-set robustification) \emph{(ours)}:} per-resource lower predictive quantile with union bound,
\[
\hat b_j = Q_{\,\alpha/m}(b_j\mid D),\qquad j=1,\dots,m,
\]
so that $P_{\theta\mid D}(\text{all constraints hold})\gtrsim 1-\alpha$ under the union bound logic.
\item \textbf{PS (Posterior-Scenario) \emph{(ours)}:} draw $N_{\text{scen}}=300$ i.i.d.\ posterior predictive samples of $b$ and enforce all sampled constraints, which is equivalent to $\hat b=\min_{k\le N_{\text{scen}}} b^{(k)}$ componentwise.
\item \textbf{FPQ (Frequentist predictive quantile):} OLS fit with Student-$t$ predictive quantiles, again using $\alpha/m$ per constraint.
\item \textbf{RB (Robust-box heuristic):} $\hat b_j=\mu_j - z_{1-\alpha/m}\,s_j$ using a Normal $z$-score heuristic.
\end{itemize}
In all cases, the LP \eqref{eq:sim_lp} is solved by a standard convex solver, and only solutions with status \texttt{optimal} (or \texttt{optimal\_inaccurate}) are retained.

\paragraph{Metrics.}
For a computed decision $\hat x$, we estimate:
\begin{enumerate}
\item \textbf{Profit} $p^\top \hat x$.
\item \textbf{True out-of-sample violation probability}
\[
\widehat V_{\text{true}}(\hat x)=\frac{1}{M}\sum_{\ell=1}^M \mathbf{1}\!\left\{A\hat x \nleq b^{(\ell)}_{\text{true}}\right\},
\]
using $M=5000$ fresh samples from the \emph{true} data-generating distribution.
\item \textbf{Posterior violation estimate and a conservative posterior certificate:}
$\widehat V_{\text{post}}(\hat x)$ estimated from $M_{\text{cert}}=5000$ posterior predictive draws, and the Clopper--Pearson one-sided 95\% upper bound $\widehat V_{\text{post,UB95}}(\hat x)$.
\end{enumerate}
We run $60$ trials for each target $\alpha\in\{0.01,0.05,0.10\}$ (total $180$ LP instances per method).

\subsection{Aggregate results and discussion}
\label{subsec:sim_results}

\paragraph{Main quantitative summaries.}
Table~\ref{tab:sim_by_alpha} reports the mean (and standard deviation) of profit and the mean (and standard deviation) of true violation probability by method and target risk $\alpha$.
Table~\ref{tab:sim_overall} aggregates across all $\alpha$.

\begin{table}[h!]
\centering
\caption{Simulation summary by target risk level $\alpha$ (60 trials per $\alpha$).
Reported are mean profit and its standard deviation, mean true violation probability and its standard deviation,
and posterior violation diagnostics (Monte Carlo mean) including a conservative 95\% upper bound.}
\label{tab:sim_by_alpha}
\small
\begin{tabular}{@{}ccccccccc@{}}
\toprule
$\alpha$ & Method & $n$ & Profit mean & Profit sd & $\widehat V_{\text{true}}$ mean & sd &
$\widehat V_{\text{post}}$ mean & $\widehat V_{\text{post,UB95}}$ mean \\
\midrule
0.0100 & CR  & 60 & 556.2547 &  94.8090 & 0.0073 & 0.0080 & 0.0048 & 0.0067 \\
0.0100 & FPQ & 60 & 543.7832 &  96.5516 & 0.0047 & 0.0057 & 0.0028 & 0.0043 \\
0.0100 & PM  & 60 & 759.5913 &  79.4713 & 0.9051 & 0.0859 & 0.9121 & 0.9184 \\
0.0100 & PS  & 60 & 558.3245 & 101.4542 & 0.0127 & 0.0125 & 0.0091 & 0.0115 \\
0.0100 & RB  & 60 & 561.7612 &  94.2242 & 0.0090 & 0.0097 & 0.0061 & 0.0082 \\
\midrule
0.0500 & CR  & 60 & 582.6296 &  80.3209 & 0.0302 & 0.0169 & 0.0239 & 0.0277 \\
0.0500 & FPQ & 60 & 572.3997 &  80.8666 & 0.0210 & 0.0123 & 0.0165 & 0.0198 \\
0.0500 & PM  & 60 & 751.3455 &  77.3139 & 0.9291 & 0.0522 & 0.9309 & 0.9365 \\
0.0500 & PS  & 60 & 551.4119 &  87.7301 & 0.0140 & 0.0108 & 0.0117 & 0.0144 \\
0.0500 & RB  & 60 & 585.8197 &  80.2025 & 0.0338 & 0.0184 & 0.0270 & 0.0311 \\
\midrule
0.1000 & CR  & 60 & 622.2889 &  79.5292 & 0.0557 & 0.0247 & 0.0468 & 0.0520 \\
0.1000 & FPQ & 60 & 613.4360 &  80.6589 & 0.0425 & 0.0206 & 0.0343 & 0.0388 \\
0.1000 & PM  & 60 & 766.1454 &  69.1882 & 0.9162 & 0.0508 & 0.9207 & 0.9267 \\
0.1000 & PS  & 60 & 567.8640 &  80.7549 & 0.0125 & 0.0108 & 0.0094 & 0.0119 \\
0.1000 & RB  & 60 & 624.6181 &  79.3269 & 0.0604 & 0.0262 & 0.0511 & 0.0565 \\
\bottomrule
\end{tabular}
\end{table}

\begin{table}[h!]
\centering
\caption{Overall summary across all $\alpha\in\{0.01,0.05,0.10\}$ (180 runs per method).}
\label{tab:sim_overall}
\small
\begin{tabular}{@{}ccccccc@{}}
\toprule
Method & $n$ & Profit mean & Profit sd & $\widehat V_{\text{true}}$ mean &
$\widehat V_{\text{post}}$ mean & $\widehat V_{\text{post,UB95}}$ mean \\
\midrule
CR  & 180 & 587.0577 &  88.9643 & 0.0311 & 0.0252 & 0.0288 \\
FPQ & 180 & 576.5396 &  90.5220 & 0.0227 & 0.0179 & 0.0210 \\
PM  & 180 & 759.0274 &  75.2771 & 0.9168 & 0.9212 & 0.9272 \\
PS  & 180 & 559.2001 &  90.1374 & 0.0131 & 0.0101 & 0.0126 \\
RB  & 180 & 590.7330 &  88.2892 & 0.0344 & 0.0281 & 0.0319 \\
\bottomrule
\end{tabular}
\end{table}

Figure~\ref{fig:sim_profit_bars} shows mean profit by method for each $\alpha$.
Figure~\ref{fig:sim_viol_bars} shows mean true violation probability by method for each $\alpha$, with the dashed line marking the target $\alpha$.
Figure~\ref{fig:sim_calibration} plots the risk-calibration curve (achieved violation versus target $\alpha$).
Figure~\ref{fig:sim_frontier} visualizes the empirical profit--risk trade-off across all trials.
Figure~\ref{fig:sim_certificate} reports the mean conservative posterior certificate $\widehat V_{\text{post,UB95}}$ versus $\alpha$.

\begin{figure}[h!]
\centering
\begin{minipage}{0.32\textwidth}
\centering
\includegraphics[width=\linewidth]{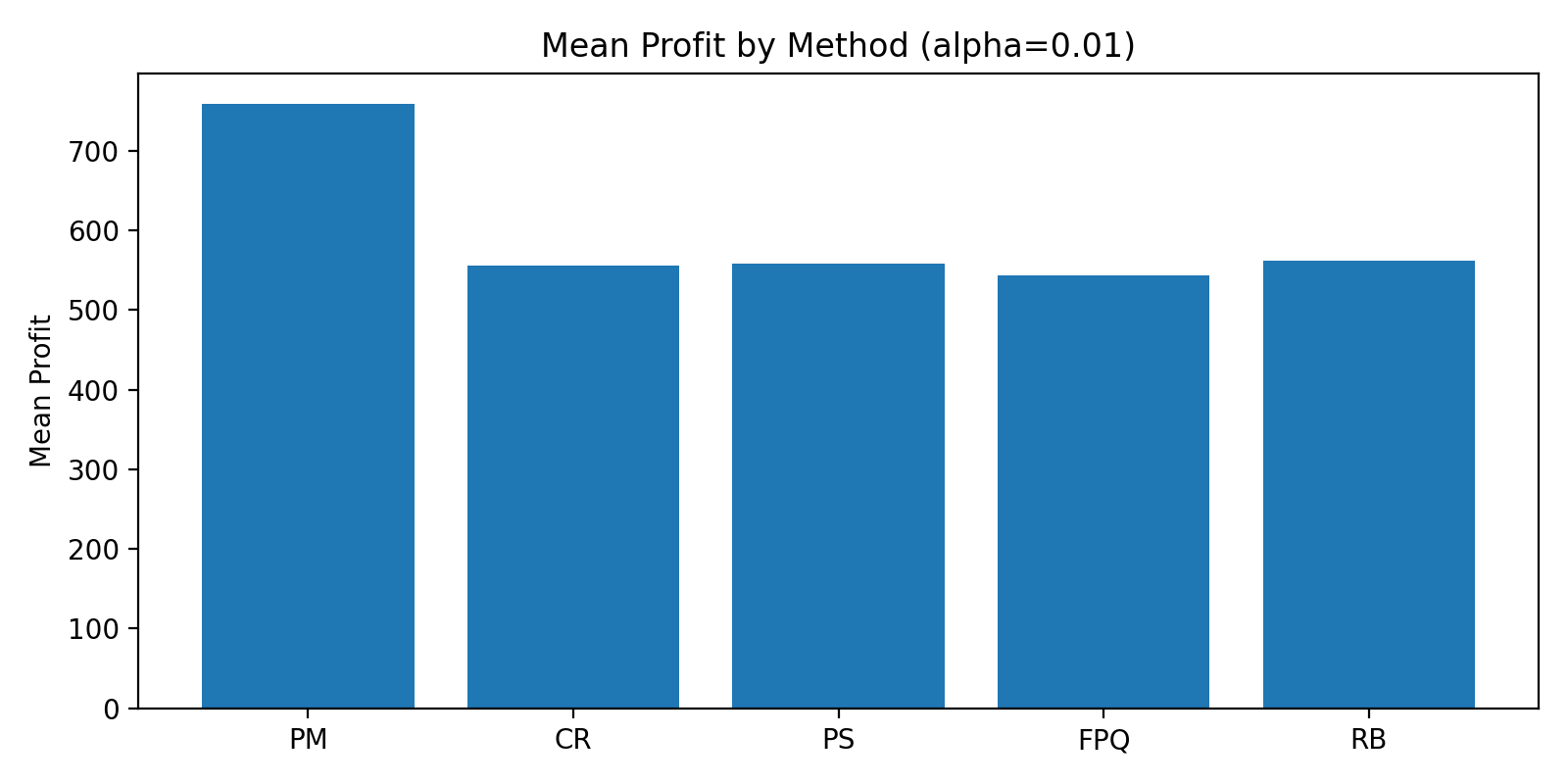}
\caption*{(a) $\alpha=0.01$}
\end{minipage}\hfill
\begin{minipage}{0.32\textwidth}
\centering
\includegraphics[width=\linewidth]{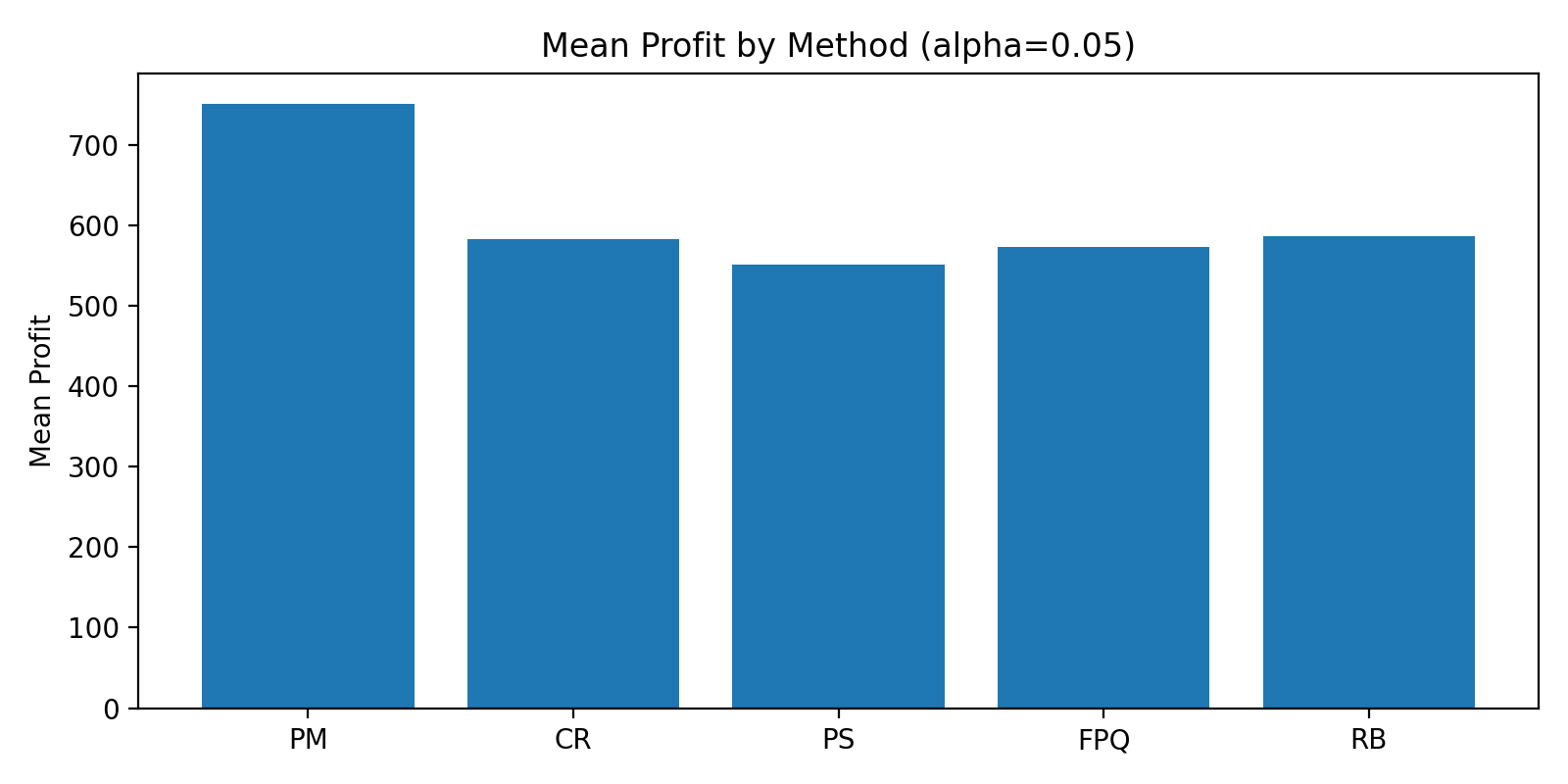}
\caption*{(b) $\alpha=0.05$}
\end{minipage}\hfill
\begin{minipage}{0.32\textwidth}
\centering
\includegraphics[width=\linewidth]{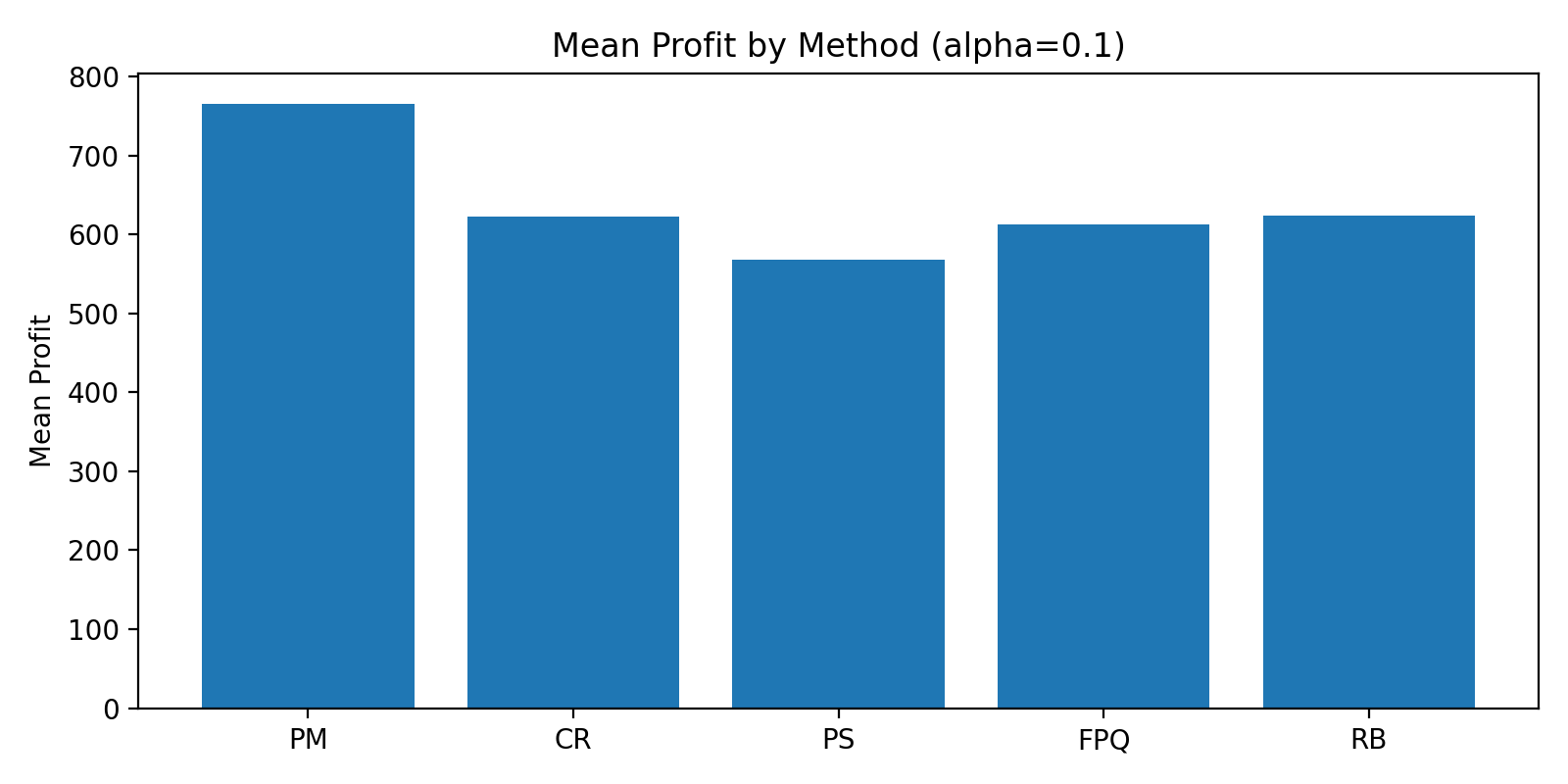}
\caption*{(c) $\alpha=0.10$}
\end{minipage}
\caption{Mean profit by method for each target risk level $\alpha$.}
\label{fig:sim_profit_bars}
\end{figure}

\begin{figure}[h!]
\centering
\begin{minipage}{0.32\textwidth}
\centering
\includegraphics[width=\linewidth]{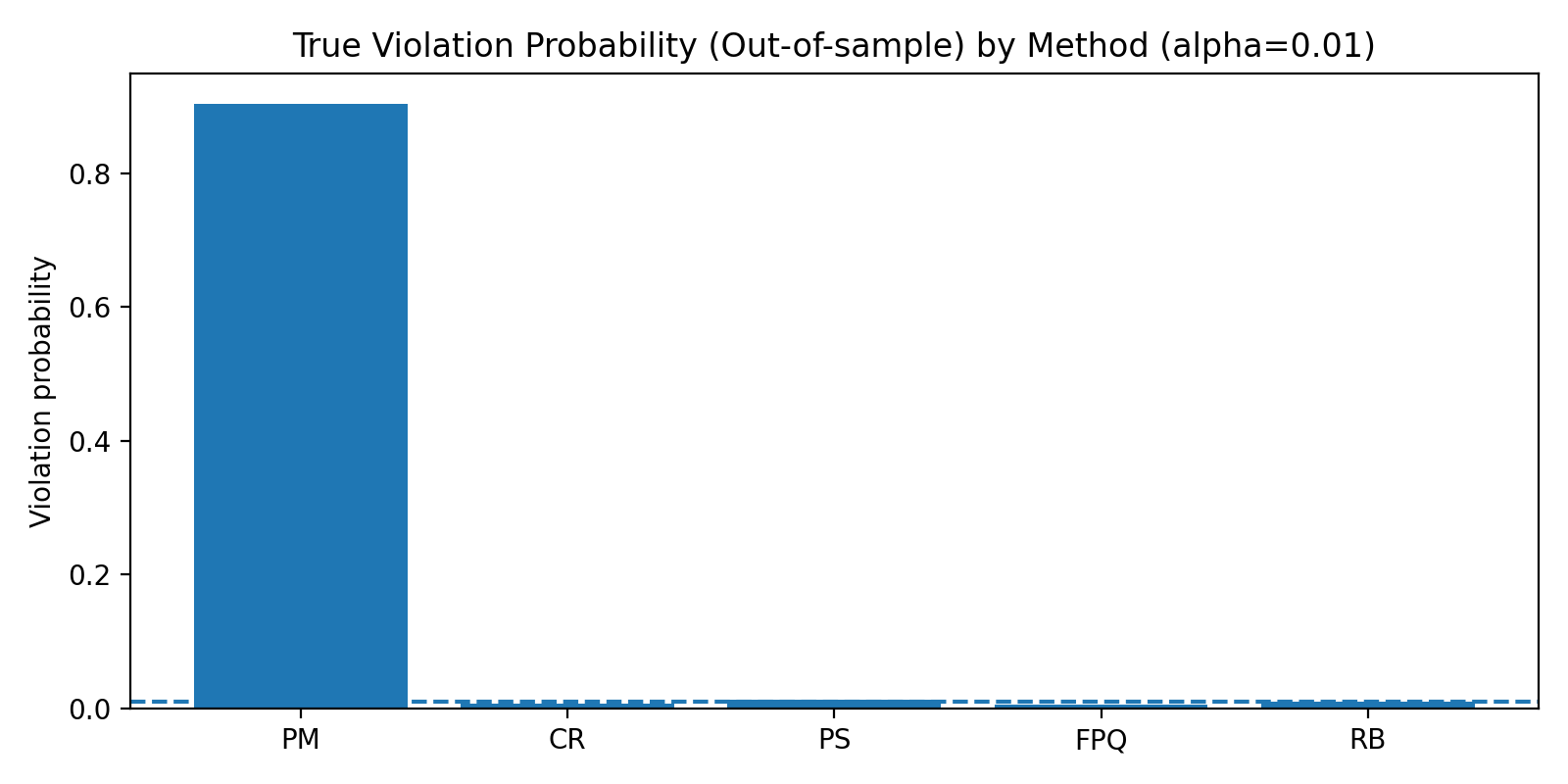}
\caption*{(a) $\alpha=0.01$}
\end{minipage}\hfill
\begin{minipage}{0.32\textwidth}
\centering
\includegraphics[width=\linewidth]{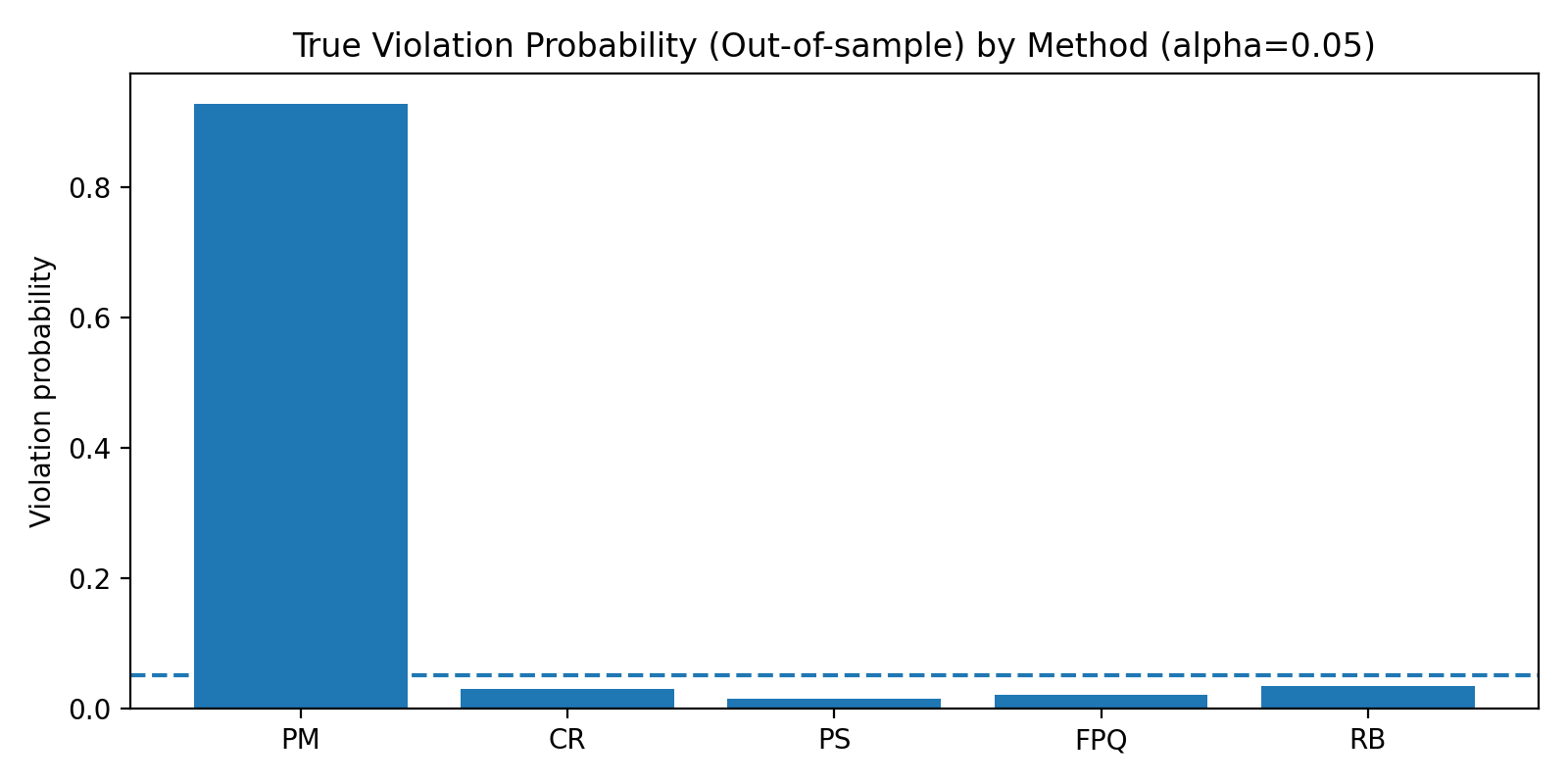}
\caption*{(b) $\alpha=0.05$}
\end{minipage}\hfill
\begin{minipage}{0.32\textwidth}
\centering
\includegraphics[width=\linewidth]{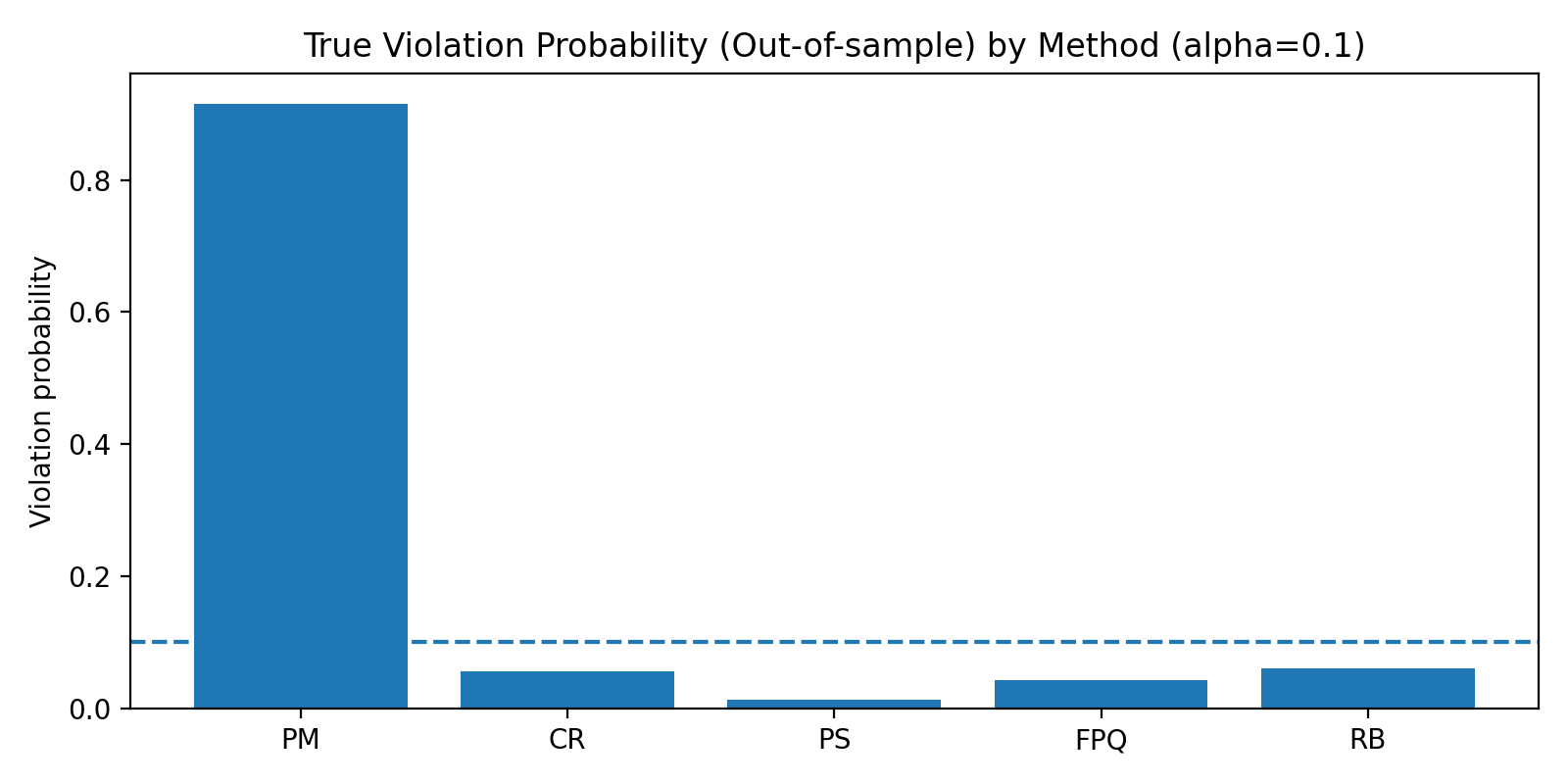}
\caption*{(c) $\alpha=0.10$}
\end{minipage}
\caption{Mean true out-of-sample violation probability $\widehat V_{\text{true}}$ by method. The dashed horizontal line marks the target $\alpha$.}
\label{fig:sim_viol_bars}
\end{figure}

\begin{figure}[h!]
\centering
\includegraphics[width=0.70\linewidth]{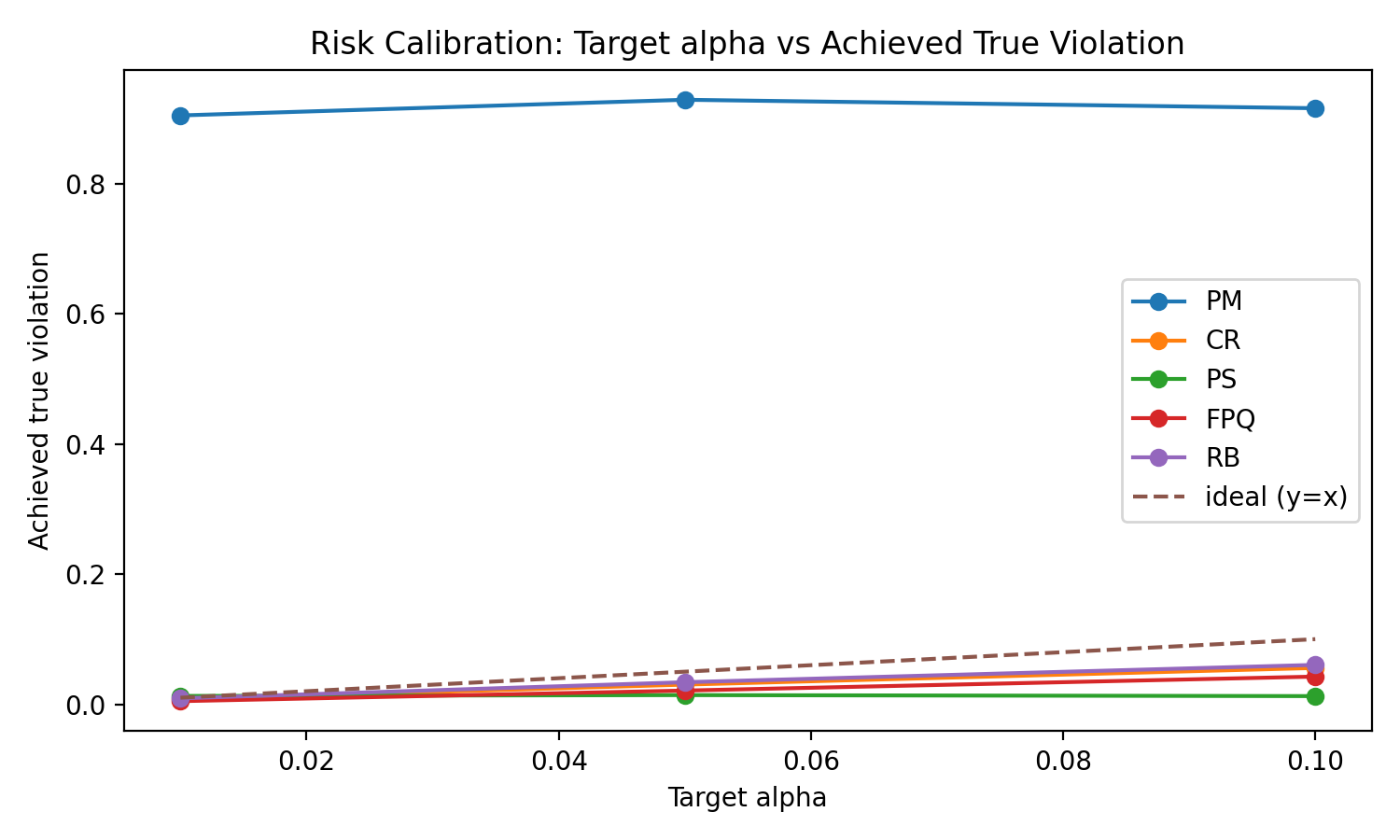}
\caption{Risk calibration: achieved true violation versus target $\alpha$ (ideal calibration corresponds to the diagonal).}
\label{fig:sim_calibration}
\end{figure}

\begin{figure}[h!]
\centering
\includegraphics[width=0.70\linewidth]{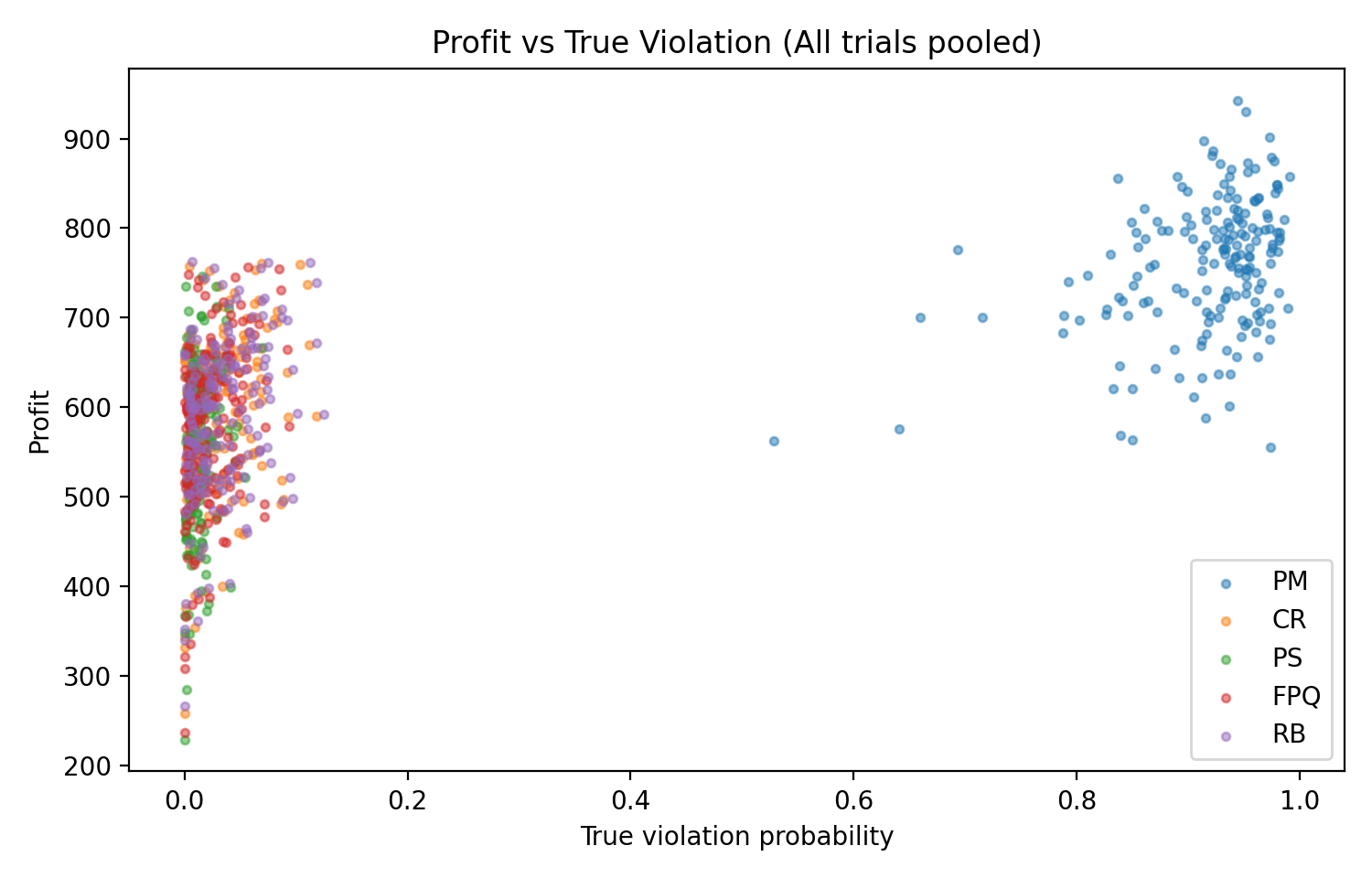}
\caption{Profit--risk trade-off across all trials pooled: scatter of profit versus true violation probability.}
\label{fig:sim_frontier}
\end{figure}

\begin{figure}[h!]
\centering
\includegraphics[width=0.70\linewidth]{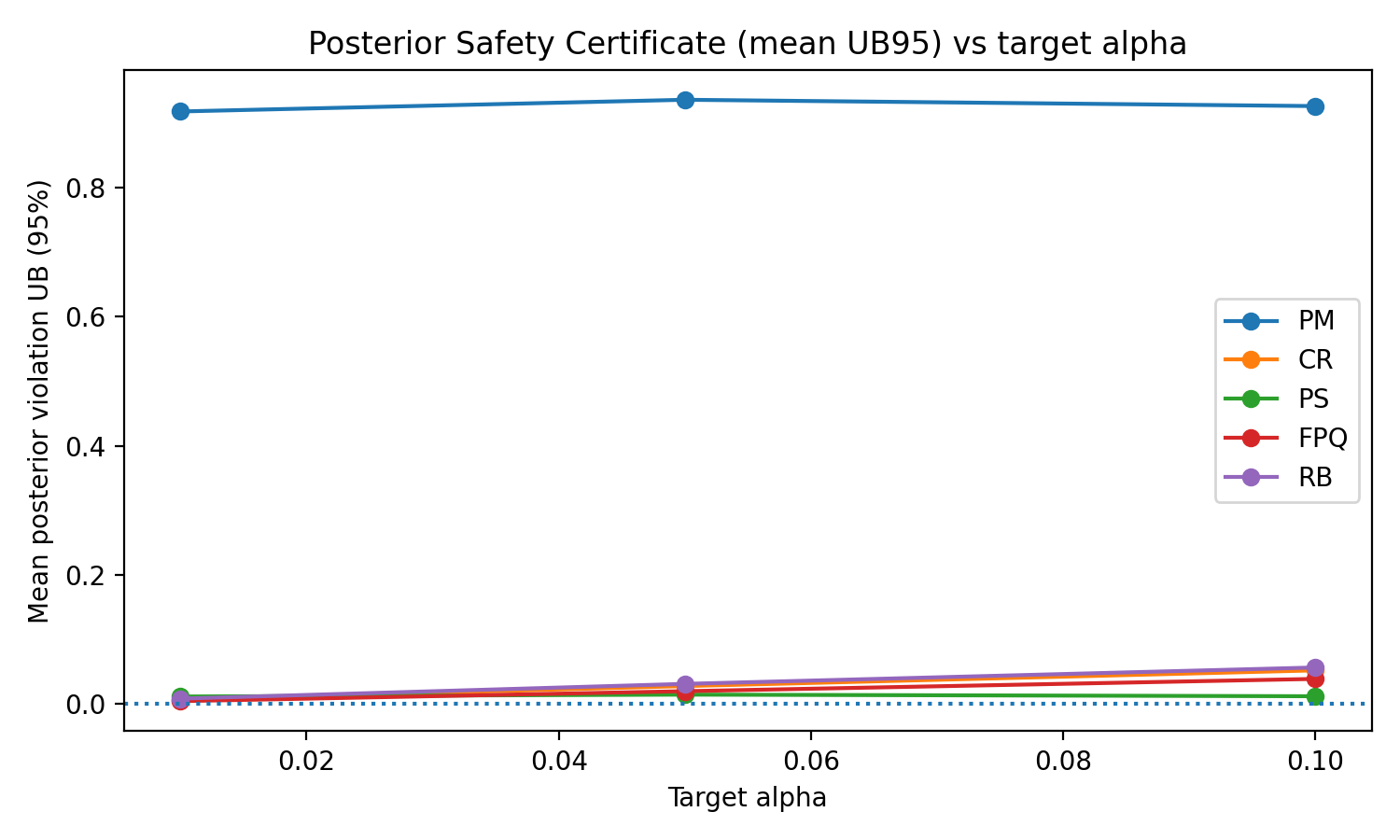}
\caption{Posterior safety certificate: mean conservative upper bound $\widehat V_{\text{post,UB95}}$ versus target $\alpha$.}
\label{fig:sim_certificate}
\end{figure}

\subsection{Key findings}
\label{rem:key_findings1}
The results in Tables~\ref{tab:sim_by_alpha}--\ref{tab:sim_overall} and Figures~\ref{fig:sim_profit_bars}--\ref{fig:sim_certificate} support the following conclusions.
\begin{enumerate}
\item \textbf{Posterior-mean plug-in is unsafe.}
PM attains the highest profit but exhibits catastrophic infeasibility: $\widehat V_{\text{true}}\approx 0.91$ across all $\alpha$, demonstrating that ignoring uncertainty can produce highly profitable yet operationally invalid decisions.

\item \textbf{The Proposed PS method is \emph{best for safety} in this simulation.}
Across all $\alpha$ levels pooled, PS has the lowest mean true violation ($0.0131$ in Table~\ref{tab:sim_overall}), and it is the lowest-violation method at $\alpha=0.05$ and $\alpha=0.10$ (Table~\ref{tab:sim_by_alpha}). Thus, under the present linear-Gaussian data-generating regime, PS delivers the strongest empirical posterior-feasibility behavior.

\item \textbf{The proposed CR method is a competitive, less conservative alternative.}
CR provides substantial safety improvements over PM and RB, while typically yielding higher profit than PS at $\alpha=0.05$ and $\alpha=0.10$ (Table~\ref{tab:sim_by_alpha}). Its achieved violation increases with $\alpha$ (Figure~\ref{fig:sim_calibration}), indicating a more ``responsive'' risk--reward trade-off than PS.

\item \textbf{A strong classical comparator exists (FPQ).}
FPQ attains the lowest mean violation at $\alpha=0.01$ (Table~\ref{tab:sim_by_alpha}) and remains competitive for other $\alpha$ values. This is unsurprising here because the data-generating process is well matched to the OLS predictive model. Importantly, The Proposed Bayesian methods (CR/PS) provide a coherent posterior-learning pipeline and explicit posterior-certification diagnostics (Figure~\ref{fig:sim_certificate}), which are central to the Proposed proposed ``posterior feasibility guarantee'' perspective.

\item \textbf{Conservatism of PS at larger $\alpha$.}
At $\alpha=0.10$, PS remains much safer than required ($\widehat V_{\text{true}}\approx 0.0125$ in Table~\ref{tab:sim_by_alpha}), indicating conservative behavior under the current choice $N_{\text{scen}}=300$ and the ``enforce-all-scenarios'' construction. This conservatism explains the lower profit of PS relative to CR/FPQ in Figure~\ref{fig:sim_profit_bars} and the profit--risk frontier in Figure~\ref{fig:sim_frontier}.
\end{enumerate}
Overall, in these experiments \textbf{our proposal is best when the primary goal is safety / posterior-feasibility} (PS), while \textbf{CR offers a competitive compromise} between safety and profit with a more calibrated response to the risk parameter $\alpha$.

\section{Real-data demonstration: safe gene-panel selection from single-cell RNA-seq}
\label{sec:realdata_pbmc}

\subsection{Dataset and scientific task}
\label{subsec:pbmc_data_task}
We demonstrate our \emph{posterior-feasibility} Bayesian LP methodology on a real genetic dataset: the PBMC3k single-cell RNA-seq dataset (processed and distributed through the \texttt{scanpy} Python library). The dataset contains $n_{\text{cells}}=2638$ peripheral-blood mononuclear cells and $p=1838$ genes, together with an 8-group cell-type annotation (Louvain-based clustering labels provided with the processed dataset).
Figure~\ref{fig:pbmc_umap} visualizes the dataset by UMAP embedding colored by cluster.

\paragraph{Goal (gene-panel selection).}
A common experimental design problem in single-cell genomics is to select a \emph{small gene panel} (e.g., for targeted assays) that remains informative across heterogeneous cell populations. We pose this as a \emph{budgeted linear program} that selects a panel of $B=30$ genes (from a candidate pool) to maximize a between-cluster discrimination score, while enforcing a \emph{posterior-feasibility constraint}: each cluster should have sufficient expected panel detectability with high posterior probability.

\subsection{Bayesian modeling of detectability and the LP formulation}
\label{subsec:pbmc_model_lp}
\paragraph{Candidate gene pool.}
From the full gene set we form a candidate pool of $K=300$ genes using a simple, transparent filter: genes with high between-cluster variability of mean expression (computed from the processed matrix). Each candidate gene $g$ is assigned a nonnegative weight $w_g$ equal to the variance of its cluster-wise mean expression. Intuitively, larger $w_g$ indicates stronger between-cluster separation.

\paragraph{Detectability as the uncertain coefficient.}
For each cluster $j\in\{1,\dots,J\}$ and candidate gene $g\in\{1,\dots,K\}$, we define a detection indicator
\[
Y_{ijg}=\mathbf{1}\{ \text{gene $g$ is detected in cell $i$ from cluster $j$} \},
\]
and let $q_{jg} = \Pr(Y_{ijg}=1)$ denote the cluster-specific detection probability.
Using a conjugate Beta--Binomial model,
\[
q_{jg} \sim \mathrm{Beta}(a_0,b_0),\qquad
\sum_{i\in \mathcal{I}_j} Y_{ijg}\mid q_{jg} \sim \mathrm{Binomial}(n_j,q_{jg}),
\]
we obtain the closed-form posterior
\[
q_{jg}\mid D \sim \mathrm{Beta}(a_0+s_{jg},\; b_0+n_j-s_{jg}),
\]
where $n_j$ is the number of cells in cluster $j$ and $s_{jg}$ is the observed number of detections for gene $g$ in cluster $j$.
We use the uniform prior $a_0=b_0=1$.

\paragraph{Decision variables and objective.}
Let $x_g\in[0,1]$ be a relaxed selection variable for candidate gene $g$; the panel-size budget is enforced by $\sum_{g=1}^K x_g \le B$.
The objective is linear:
\[
\max_{x}\ \sum_{g=1}^K w_g x_g.
\]

\paragraph{Posterior-feasibility constraint (the Proposed method: PS).}
We require each cluster to have at least $L=8$ \emph{expected detectable} panel genes under the posterior, in a scenario-based (posterior-sampled) sense. Specifically, draw $S=300$ i.i.d.\ posterior samples
$q_{jg}^{(s)}\sim \mathrm{Beta}(a_{jg},b_{jg})$ and enforce
\[
\sum_{g=1}^K q_{jg}^{(s)} x_g \ge L,\qquad \forall j=1,\dots,J,\ \forall s=1,\dots,S.
\]
This is a \emph{linear program} in $x$ once the posterior scenarios are sampled, and it is exactly the Proposed \emph{Posterior-Scenario (PS)} posterior-feasibility construction.

\paragraph{Hard panel extraction.}
After solving the relaxed LP, we form the final gene panel by taking the top $B=30$ genes by the optimized $x_g$ values.

\subsection{Empirical posterior-feasibility certification}
\label{subsec:pbmc_cert}
Beyond optimization-time scenario enforcement, we perform an independent posterior Monte Carlo certification: we draw $M_{\text{cert}}=4000$ fresh posterior samples of $(q_{jg})$ and estimate the posterior violation probability,
\[
\widehat V_{\text{post}}(\hat x) = \frac{1}{M_{\text{cert}}}\sum_{\ell=1}^{M_{\text{cert}}}
\mathbf{1}\Big\{\exists j:\ \sum_{g=1}^K q_{jg}^{(\ell)} \hat x_g < L\Big\}.
\]
We also report a conservative one-sided 95\% Clopper--Pearson upper bound $\widehat V_{\text{post,UB95}}(\hat x)$.

Table~\ref{tab:pbmc_certificate} summarizes the global certificate, and Table~\ref{tab:pbmc_cluster_detail} reports cluster-wise coverage distributions and violation rates. Figure~\ref{fig:pbmc_coverage_boxplot} visualizes the posterior coverage distributions across clusters, and Figure~\ref{fig:pbmc_violation_hist} summarizes feasible vs.\ violating posterior draws.

\subsection{Results}
\label{subsec:pbmc_results}


\begin{figure}[h!]
\centering
\includegraphics[width=0.78\linewidth]{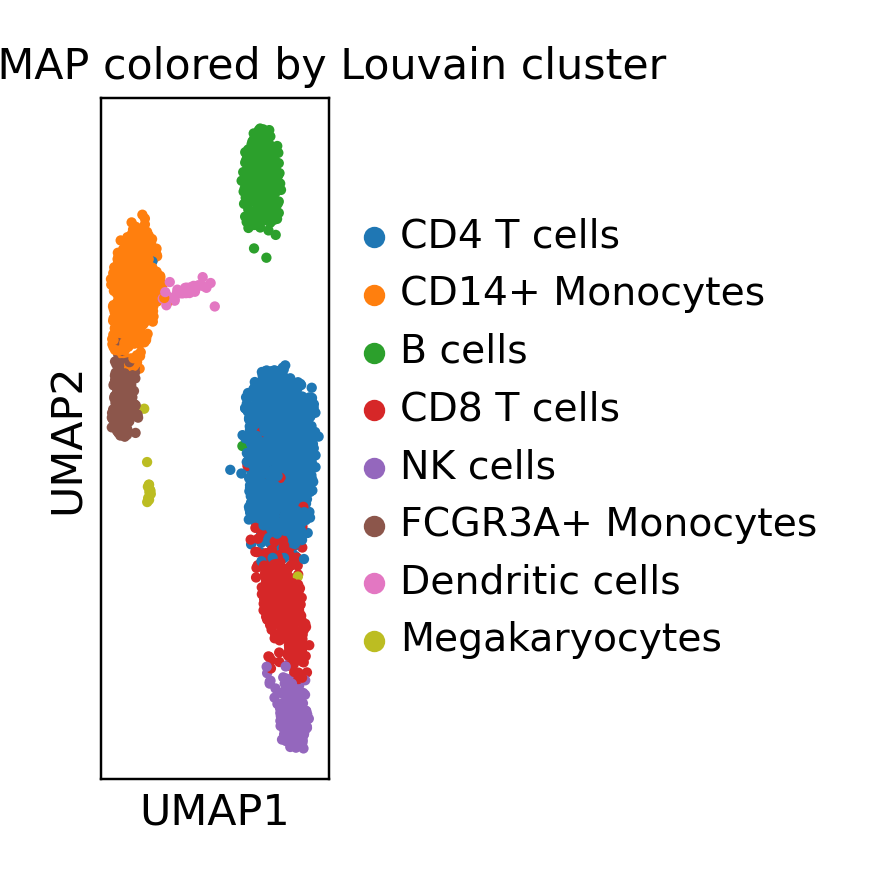}
\caption{PBMC3k UMAP embedding colored by the 8 provided cluster labels.}
\label{fig:pbmc_umap}
\end{figure}

\begin{figure}[h!]
\centering
\includegraphics[width=0.90\linewidth]{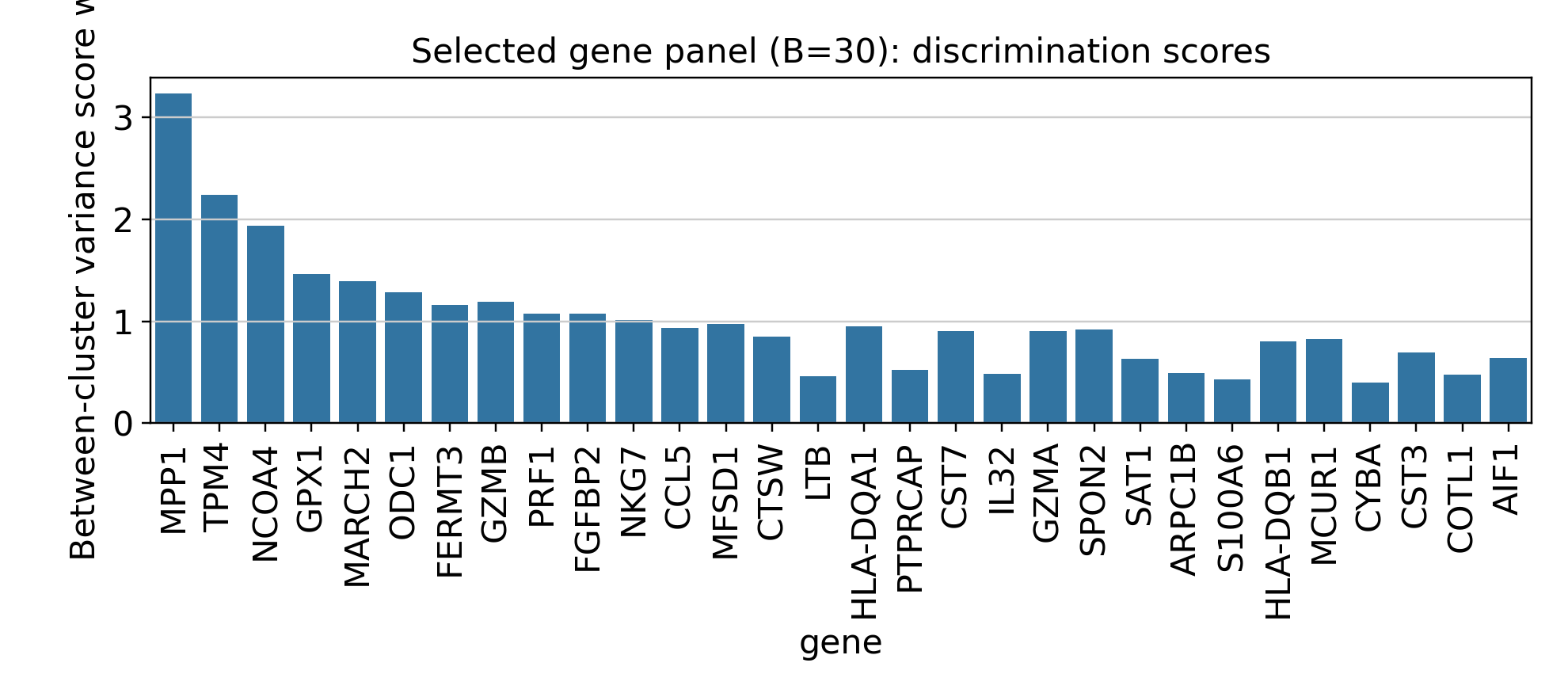}
\caption{Discrimination weights $w_g$ for the selected panel genes (larger is more between-cluster separation).}
\label{fig:pbmc_panel_scores}
\end{figure}

\begin{figure}[h!]
\centering
\includegraphics[width=0.92\linewidth]{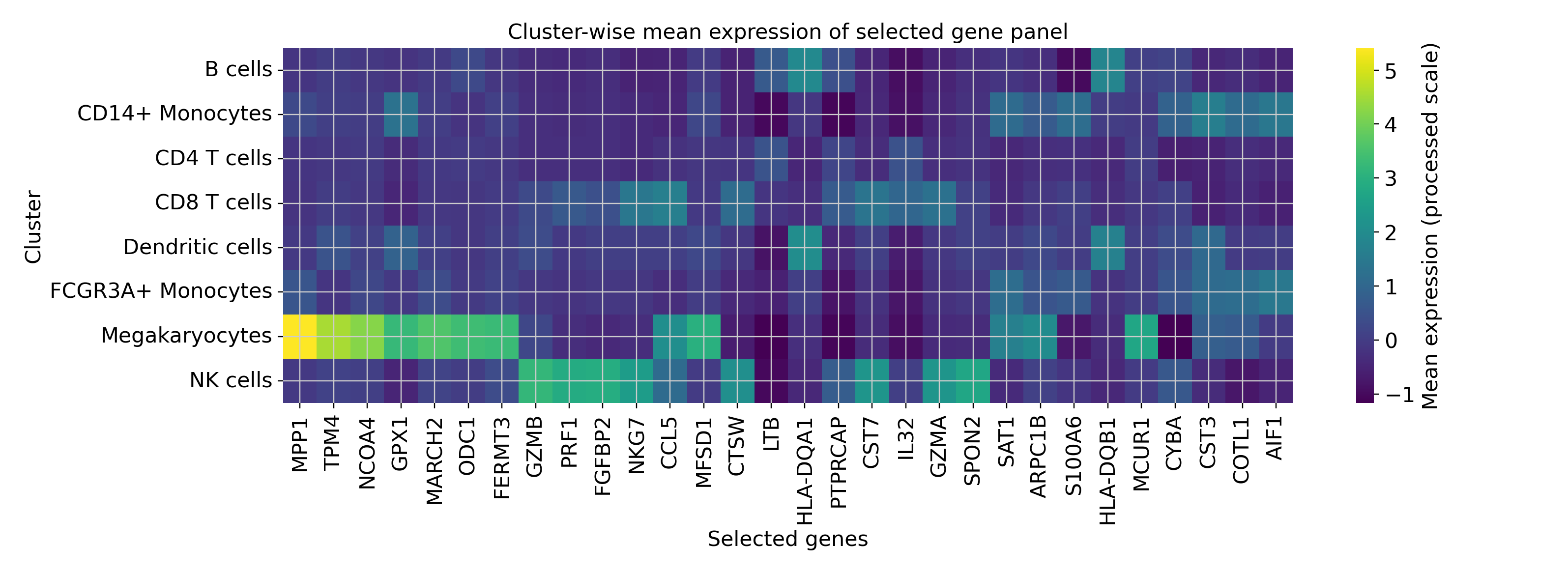}
\caption{Heatmap of cluster-wise mean expression (processed scale) for the selected $B=30$ gene panel.
This provides a compact, interpretable view of how the chosen panel separates clusters.}
\label{fig:pbmc_heatmap}
\end{figure}

\begin{figure}[h!]
\centering
\includegraphics[width=0.80\linewidth]{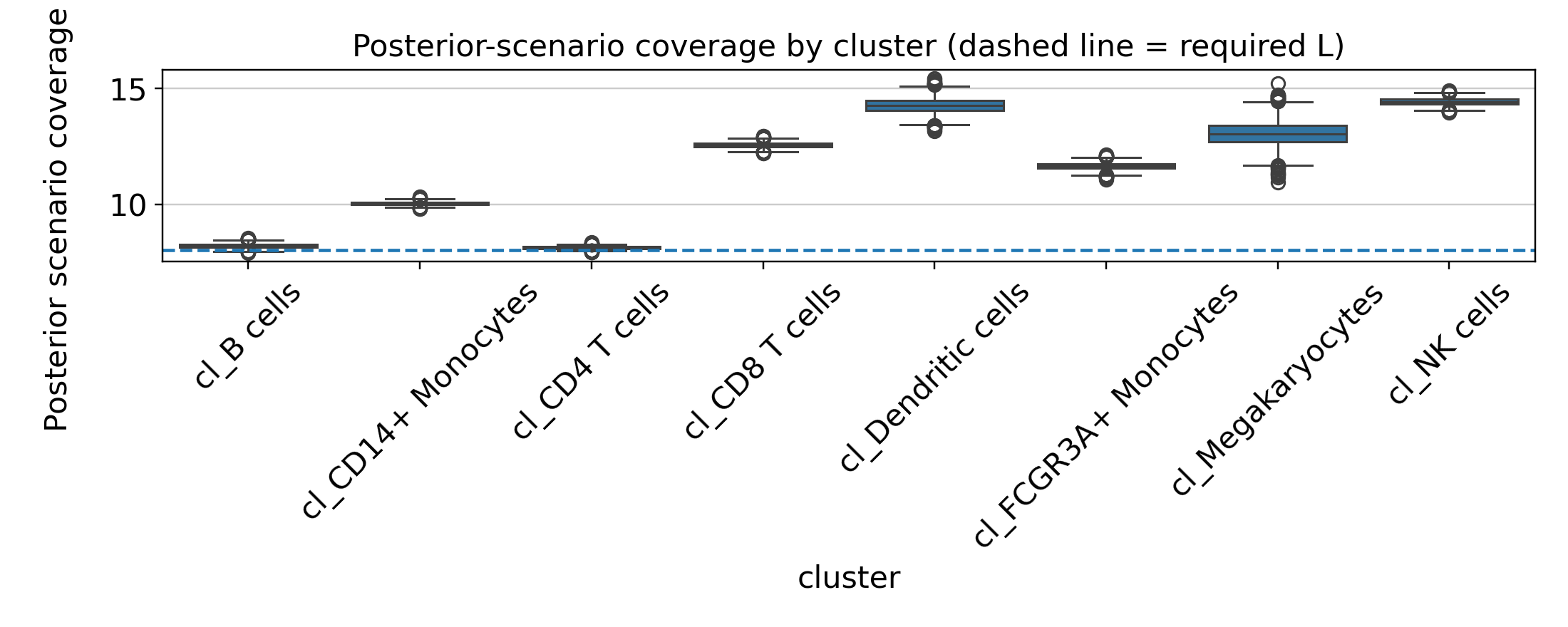}
\caption{Posterior-scenario coverage distributions by cluster: boxplots of $\sum_g q_{jg}\hat x_g$ across $M_{\text{cert}}=4000$ posterior draws.
The dashed horizontal line is the required threshold $L=8$.}
\label{fig:pbmc_coverage_boxplot}
\end{figure}

\begin{figure}[h!]
\centering
\includegraphics[width=0.55\linewidth]{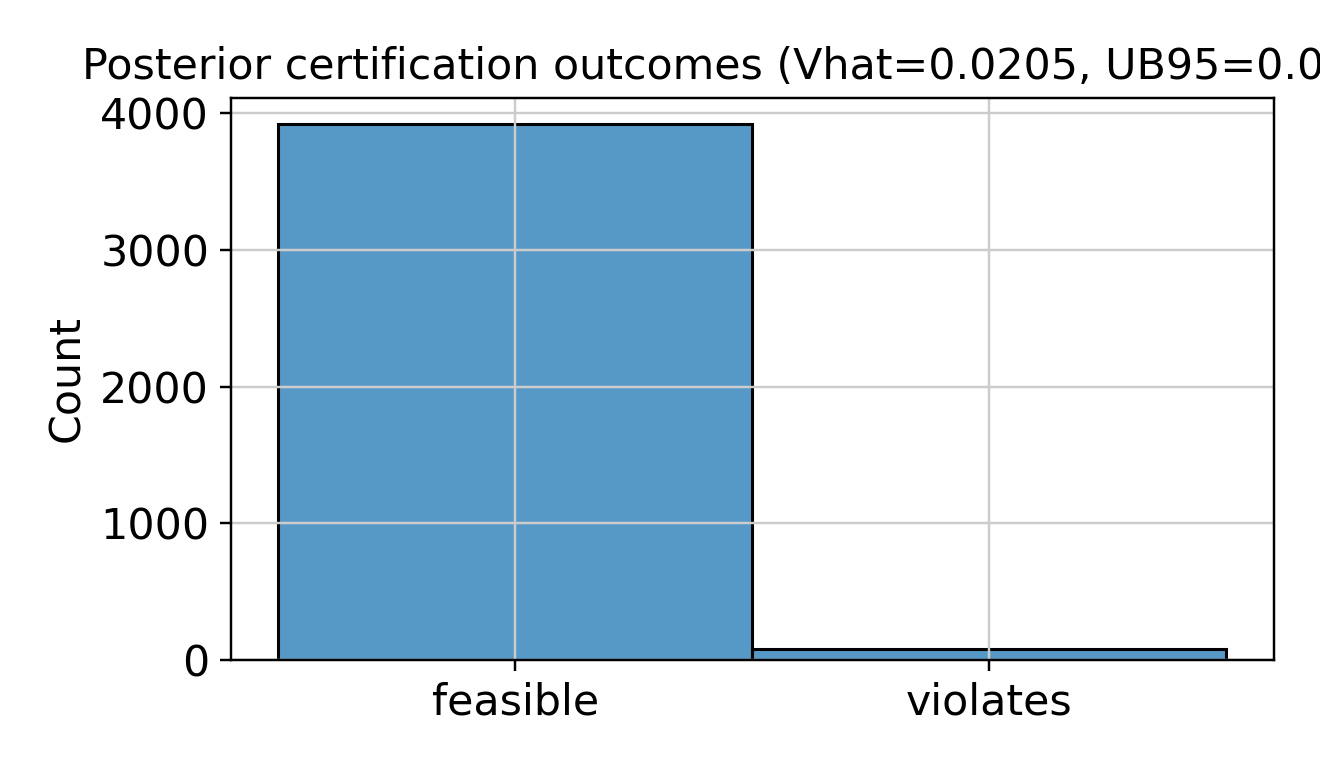}
\caption{Posterior feasibility certification outcomes (feasible vs.\ violating posterior draws),
corresponding to $\widehat V_{\text{post}}=0.0205$ and $\widehat V_{\text{post,UB95}}=0.024582$ (Table~\ref{tab:pbmc_certificate}).}
\label{fig:pbmc_violation_hist}
\end{figure}

\begin{figure}[h!]
\centering
\begin{minipage}{0.32\textwidth}\centering
\includegraphics[width=\linewidth]{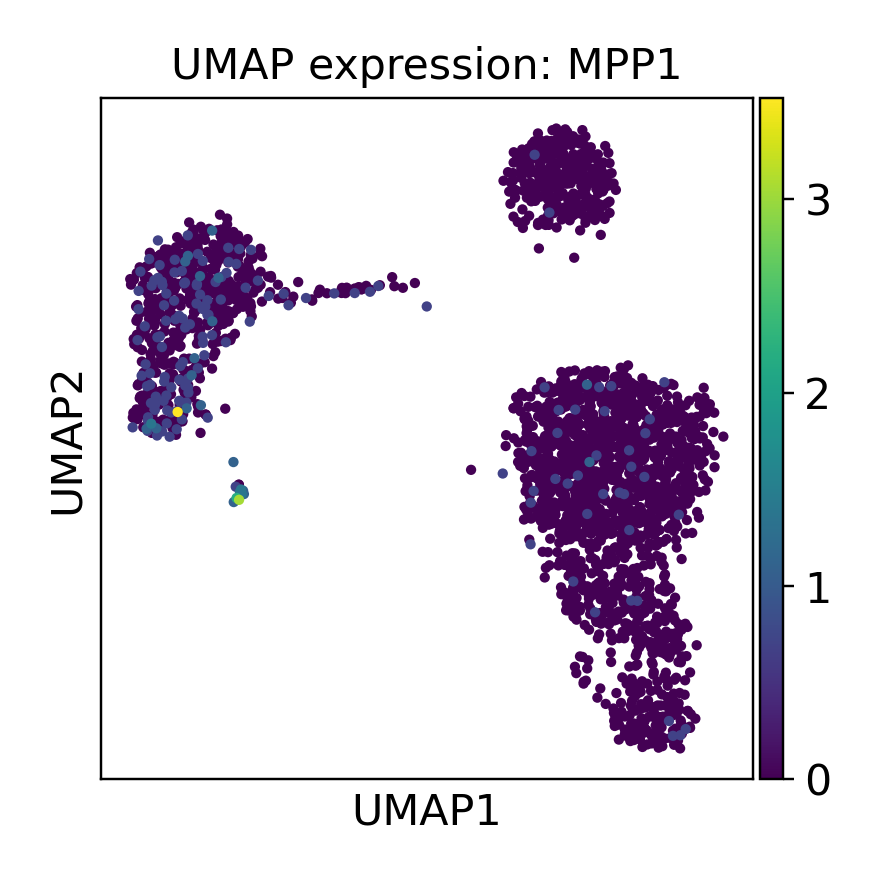}
\caption*{(a) MPP1}
\end{minipage}\hfill
\begin{minipage}{0.32\textwidth}\centering
\includegraphics[width=\linewidth]{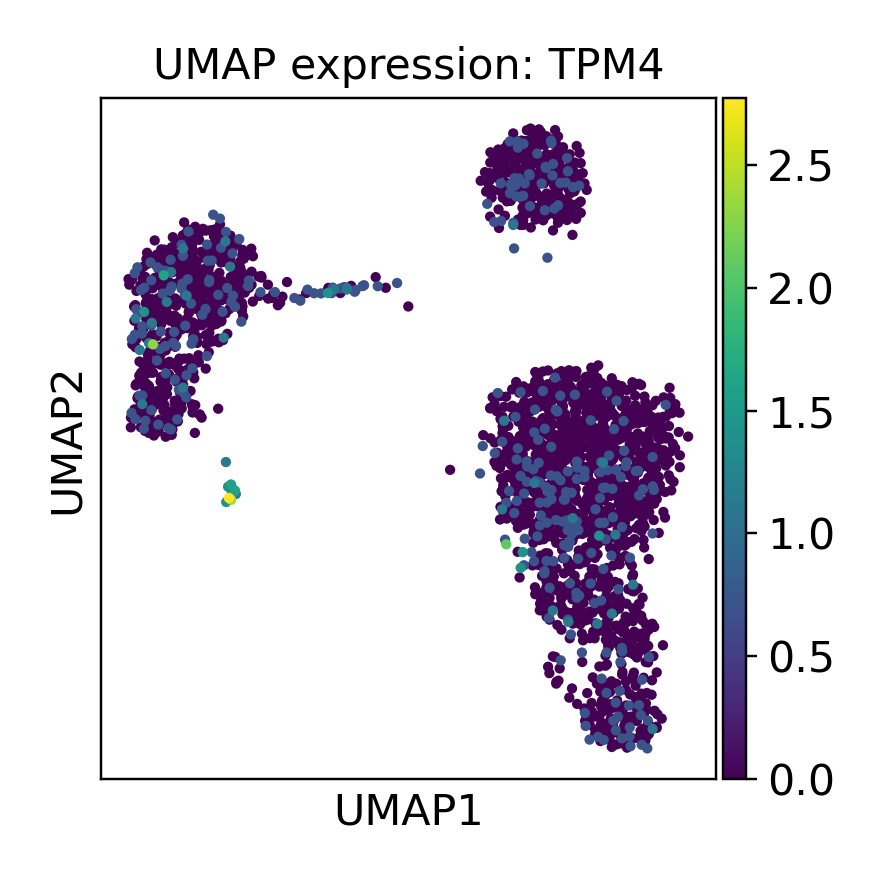}
\caption*{(b) TPM4}
\end{minipage}\hfill
\begin{minipage}{0.32\textwidth}\centering
\includegraphics[width=\linewidth]{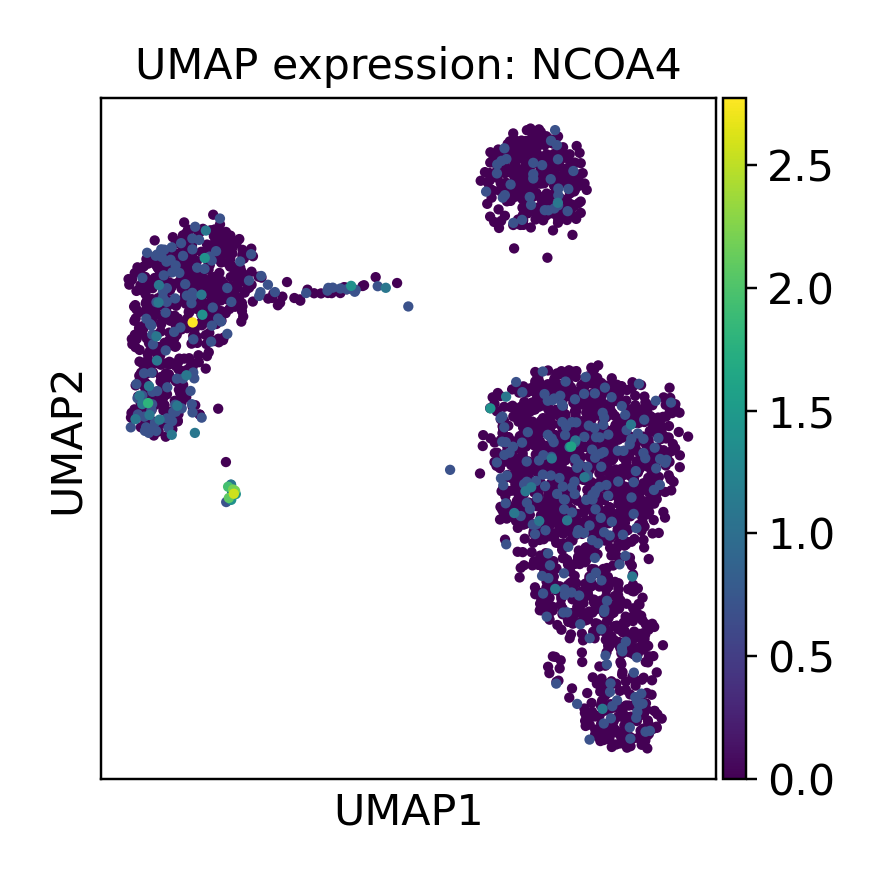}
\caption*{(c) NCOA4}
\end{minipage}

\vspace{0.6em}

\begin{minipage}{0.32\textwidth}\centering
\includegraphics[width=\linewidth]{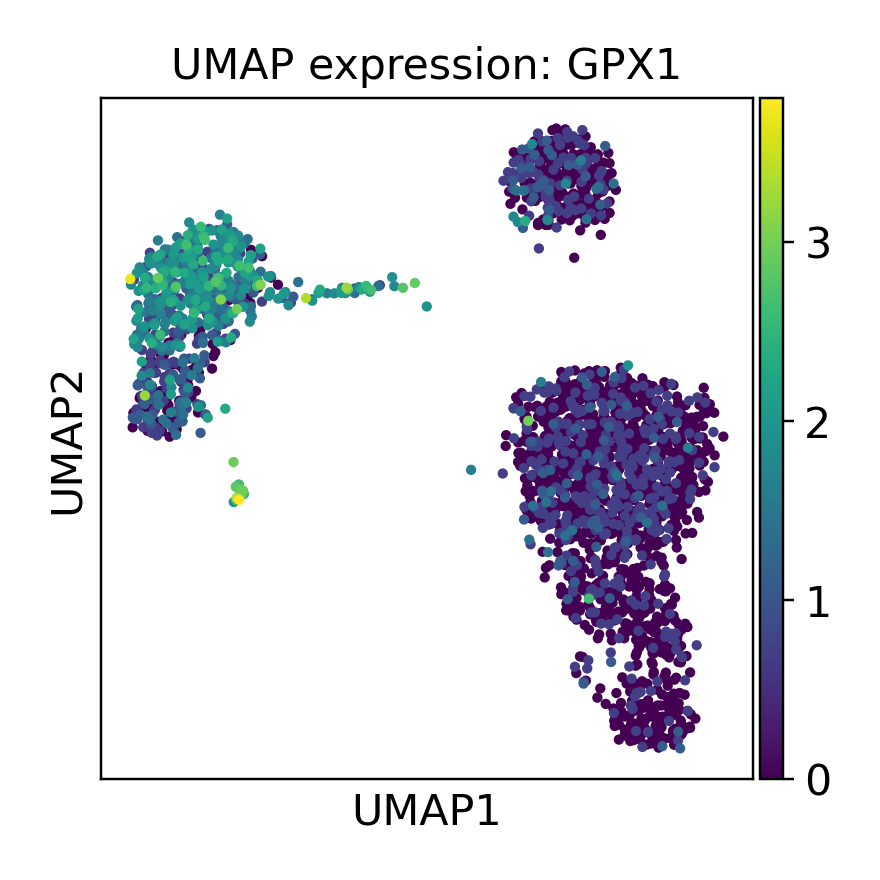}
\caption*{(d) GPX1}
\end{minipage}\hfill
\begin{minipage}{0.32\textwidth}\centering
\includegraphics[width=\linewidth]{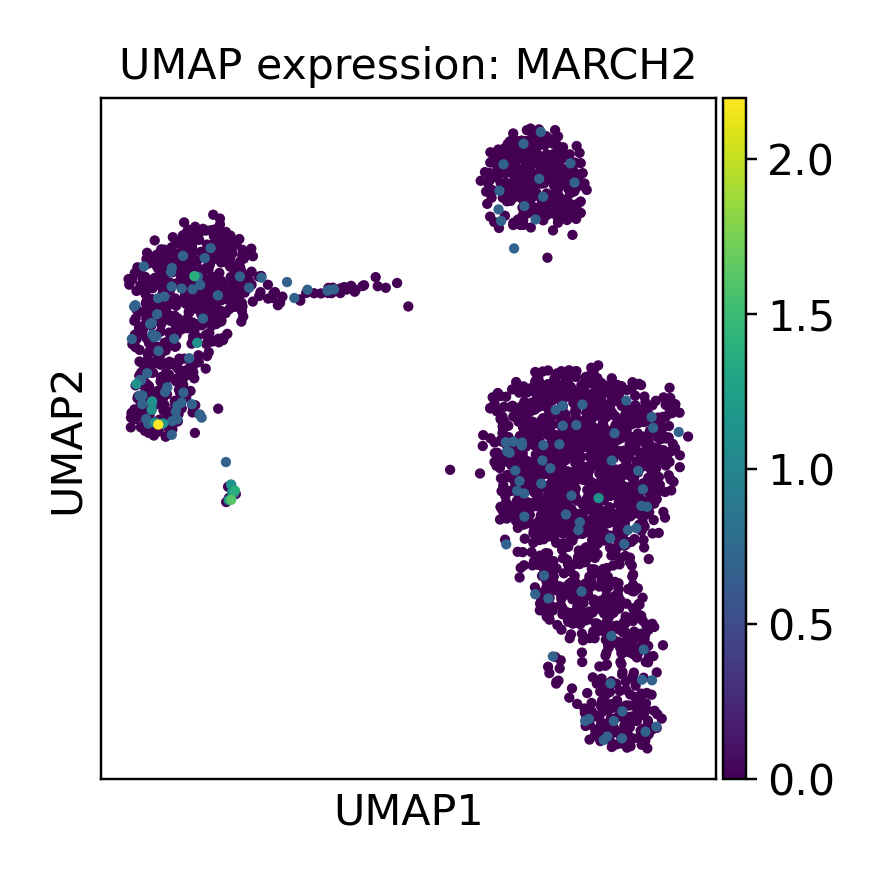}
\caption*{(e) MARCH2}
\end{minipage}\hfill
\begin{minipage}{0.32\textwidth}\centering
\includegraphics[width=\linewidth]{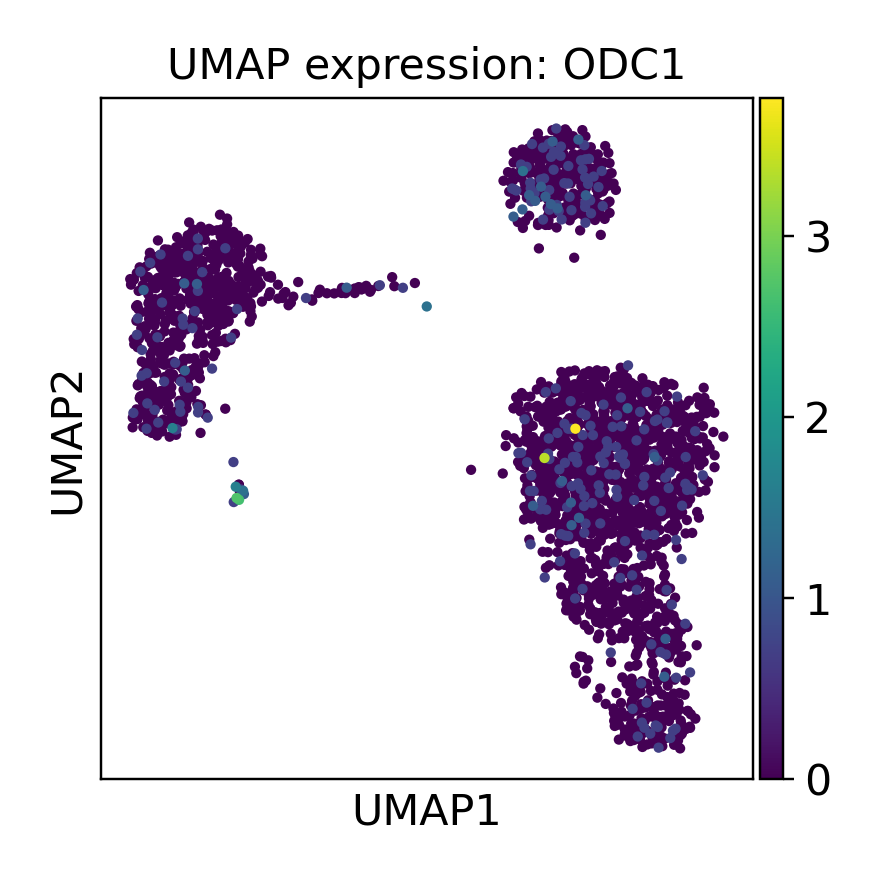}
\caption*{(f) ODC1}
\end{minipage}

\caption{UMAP feature plots for the top six selected panel genes (highest $w_g$).}
\label{fig:pbmc_topgene_umaps}
\end{figure}

\paragraph{Cluster composition.}
Table~\ref{tab:pbmc_cluster_sizes} shows the cluster sizes ($J=8$ clusters) used in our detectability model.

\begin{table}[h!]
\centering
\caption{PBMC3k cluster sizes used for modeling gene detectability.}
\label{tab:pbmc_cluster_sizes}
\small
\begin{tabular}{@{}lc@{}}
\toprule
Cluster & \# Cells \\
\midrule
B cells & 342 \\
CD14+ Monocytes & 480 \\
CD4 T cells & 1144 \\
CD8 T cells & 316 \\
Dendritic cells & 37 \\
FCGR3A+ Monocytes & 150 \\
Megakaryocytes & 15 \\
NK cells & 154 \\
\bottomrule
\end{tabular}
\end{table}

\paragraph{Selected gene panel.}
Table~\ref{tab:pbmc_panel} lists the selected $B=30$ genes, their relaxed LP selection values ($x_g$), and discrimination scores ($w_g$). The bar chart in Figure~\ref{fig:pbmc_panel_scores} visualizes the $w_g$ values for the selected panel.

\begin{table}[h!]
\centering
\caption{Selected gene panel ($B=30$) from the posterior-scenario Bayesian LP.}
\label{tab:pbmc_panel}
\scriptsize
\begin{tabular}{@{}rllr@{}}
\toprule
Rank & Gene & $x_g$ (relaxed) & $w_g$ \\
\midrule
 1 & MPP1     & 1.000000 & 3.230263 \\
 2 & TPM4     & 1.000000 & 2.240738 \\
 3 & NCOA4    & 1.000000 & 1.936093 \\
 4 & GPX1     & 1.000000 & 1.462955 \\
 5 & MARCH2   & 1.000000 & 1.389063 \\
 6 & ODC1     & 1.000000 & 1.285478 \\
 7 & FERMT3   & 1.000000 & 1.159700 \\
 8 & GZMB     & 1.000000 & 1.190489 \\
 9 & PRF1     & 1.000000 & 1.073113 \\
10 & FGFBP2   & 1.000000 & 1.073560 \\
11 & NKG7     & 1.000000 & 1.010056 \\
12 & CCL5     & 1.000000 & 0.931688 \\
13 & MFSD1    & 1.000000 & 0.972855 \\
14 & CTSW     & 1.000000 & 0.845375 \\
15 & LTB      & 1.000000 & 0.461896 \\
16 & HLA-DQA1 & 1.000000 & 0.949899 \\
17 & PTPRCAP  & 1.000000 & 0.526779 \\
18 & CST7     & 1.000000 & 0.906956 \\
19 & IL32     & 1.000000 & 0.485118 \\
20 & GZMA     & 1.000000 & 0.901278 \\
21 & SPON2    & 1.000000 & 0.915694 \\
22 & SAT1     & 1.000000 & 0.634532 \\
23 & ARPC1B   & 1.000000 & 0.488807 \\
24 & S100A6   & 1.000000 & 0.426665 \\
25 & HLA-DQB1 & 1.000000 & 0.803138 \\
26 & MCUR1    & 1.000000 & 0.825031 \\
27 & CYBA     & 1.000000 & 0.401687 \\
28 & CST3     & 1.000000 & 0.695318 \\
29 & COTL1    & 0.957915 & 0.477773 \\
30 & AIF1     & 0.801994 & 0.641515 \\
\bottomrule
\end{tabular}
\end{table}

\paragraph{Posterior-feasibility certificate.}
Table~\ref{tab:pbmc_certificate} reports the estimated posterior violation probability and a conservative 95\% upper bound.
In this run, the posterior violation estimate is $\widehat V_{\text{post}}=0.0205$ and the conservative upper bound is $\widehat V_{\text{post,UB95}}=0.024582$, demonstrating \emph{strong posterior-feasibility} for the chosen threshold $L=8$.

\begin{table}[h!]
\centering
\caption{Global posterior-feasibility certificate for the selected panel ($B=30$, $L=8$).}
\label{tab:pbmc_certificate}
\small
\begin{tabular}{@{}ccccccc@{}}
\toprule
$\alpha$ (intent) & $B$ & $L$ & $S$ (opt) & $M_{\text{cert}}$ &
$\widehat V_{\text{post}}$ & $\widehat V_{\text{post,UB95}}$ \\
\midrule
0.05 & 30 & 8.0 & 300 & 4000 & 0.0205 & 0.024582 \\
\bottomrule
\end{tabular}
\end{table}

\paragraph{Cluster-wise safety margins.}
Table~\ref{tab:pbmc_cluster_detail} shows the posterior distribution of each cluster’s coverage
$\sum_g q_{jg} \hat x_g$ (mean and selected quantiles) and its cluster-specific violation rate
$\Pr(\sum_g q_{jg}\hat x_g < L \mid D)$ estimated from the $M_{\text{cert}}$ posterior draws. The most binding clusters are \texttt{CD4 T cells} and \texttt{B cells}, whose mean coverages (8.123 and 8.208) lie close to the threshold $L=8$ and whose violation rates are nonzero (0.0120 and 0.00875, respectively), while the remaining clusters have substantial slack.

\begin{table}[h!]
\centering
\caption{Cluster-wise posterior coverage summaries for the selected panel (coverage $=\sum_g q_{jg}\hat x_g$; threshold $L=8$).}
\label{tab:pbmc_cluster_detail}
\small
\begin{tabular}{@{}lccccc@{}}
\toprule
Cluster & Mean & $5\%$ & Median & $95\%$ & Violation rate \\
\midrule
B cells & 8.207970 & 8.053454 & 8.207325 & 8.366556 & 0.00875 \\
CD14+ Monocytes & 10.043267 & 9.925323 & 10.043179 & 10.157165 & 0.00000 \\
CD4 T cells & 8.123282 & 8.034384 & 8.125445 & 8.211415 & 0.01200 \\
CD8 T cells & 12.563760 & 12.387628 & 12.563310 & 12.738323 & 0.00000 \\
Dendritic cells & 14.274472 & 13.761980 & 14.272147 & 14.807459 & 0.00000 \\
FCGR3A+ Monocytes & 11.642840 & 11.402140 & 11.641145 & 11.889237 & 0.00000 \\
Megakaryocytes & 13.059786 & 12.209755 & 13.053074 & 13.880318 & 0.00000 \\
NK cells & 14.441874 & 14.210318 & 14.440666 & 14.669766 & 0.00000 \\
\bottomrule
\end{tabular}
\end{table}

\subsection{Key findings}
\label{rem:pbmc_keyfindings}
This real-data experiment highlights why \emph{posterior-feasibility} Bayesian LP is practically valuable in genomics settings:
\begin{enumerate}
\item \textbf{A transparent optimization target with an explicit uncertainty-aware safety constraint.}
The objective (between-cluster separability) is interpretable and directly linked to the visualization in Figure~\ref{fig:pbmc_heatmap}, while feasibility is expressed as a concrete per-cluster detectability requirement $L=8$ (Figures~\ref{fig:pbmc_coverage_boxplot}--\ref{fig:pbmc_violation_hist}).

\item \textbf{Posterior-feasibility is quantitatively certified.}
The selected panel achieves a small estimated posterior violation probability ($\widehat V_{\text{post}}=0.0205$) with a conservative 95\% upper bound of $0.024582$ (Table~\ref{tab:pbmc_certificate}). This is an actionable, reproducible statement about uncertainty and feasibility, not merely a point estimate.

\item \textbf{The method adapts to heterogeneous clusters.}
Cluster-wise summaries (Table~\ref{tab:pbmc_cluster_detail}) show that \texttt{CD4 T cells} and \texttt{B cells} are the most binding groups (mean coverage close to $L$), while other clusters have substantial slack. This behavior is desirable: the LP automatically allocates panel capacity to maintain feasibility where it is hardest.

\item \textbf{The output is biologically interpretable and auditable.}
The final gene panel (Table~\ref{tab:pbmc_panel}) and gene-level spatial expression patterns (Figure~\ref{fig:pbmc_topgene_umaps}) provide an auditable bridge from optimization to scientific interpretation, which is often essential for genomics decision pipelines.

\end{enumerate}
The PBMC3k case study demonstrates that our posterior-scenario Bayesian LP produces a \emph{compact, interpretable gene panel} together with \emph{explicit posterior-feasibility guarantees and diagnostics}, a combination that is difficult to obtain from purely plug-in or purely heuristic selection approaches.
%

\section{Discussion}
The framework developed in this paper is motivated by a simple yet important observation: in many contemporary linear programming applications, uncertainty is not merely an exogenous perturbation around a fixed model, but rather the output of statistical learning. Once this is recognized, it becomes natural to ask that feasibility be assessed with respect to the learned posterior distribution rather than a fixed nominal estimate alone. The notion of posterior feasibility provides precisely such a bridge.

A central conceptual advantage of the proposed approach is that it separates two sources of reasoning that are often conflated in practice. The first concerns \emph{learning from data}: how uncertain are the coefficients after observing data? The second concerns \emph{decision protection}: how conservative should the resulting decision be in the face of that uncertainty? Bayesian updating addresses the former, while the feasible-set constructions developed here address the latter. This separation leads to a more transparent decision pipeline, in which the role of prior information, observed data, posterior uncertainty, and optimization conservatism can each be explicitly examined.

The two computational strategies studied in this paper serve different purposes. Credible-set robustification is attractive when one seeks a deterministic optimization problem with direct worst-case interpretation over a posterior high-mass region. It is conceptually simple and computationally stable, but may become conservative when the credible region is large or high-dimensional. By contrast, the posterior-scenario construction is more flexible and often easier to implement in complex models where posterior sampling is available but analytic credible regions are difficult to characterize. Its principal strength is that it aligns naturally with modern Bayesian computation, though its practical performance depends on the number of sampled scenarios and may become conservative if too many extreme posterior realizations are enforced simultaneously.

The simulation study clearly highlights these trade-offs. The posterior-scenario approach provided the strongest empirical safety performance, whereas credible-set robustification offered a more balanced safety--profit compromise. At the same time, the experiments show that no single method dominates across all metrics. In particular, when the data-generating mechanism is closely aligned with a classical predictive model, a carefully specified frequentist predictive-quantile method can remain highly competitive. This is an important point: the contribution of the present work is not that Bayesian methods must always dominate numerically, but that they provide a coherent framework in which uncertainty learned from data is propagated into optimization and certified in an interpretable way.

The real-data study further illustrates the framework's value beyond synthetic benchmarks. In the PBMC3k single-cell RNA-seq application, the method produced a compact, interpretable gene panel, along with explicit posterior-feasibility diagnostics. This is especially important in scientific applications, where decision quality must often be communicated not only through an objective value but also through transparent uncertainty summaries and cluster-specific safety margins. The resulting analysis demonstrates that Bayesian linear programming can support scientifically meaningful decisions while retaining formal uncertainty-aware structure.

Several limitations should also be acknowledged. First, the posterior guarantees in this paper are conditional on the assumed statistical model; as with all Bayesian procedures, model misspecification can affect the quality of resulting decisions. Second, the current analysis focuses primarily on single-stage linear programs and does not fully address recourse, adaptivity, or sequential decision-making. Third, while the scenario-based guarantees are theoretically appealing, their practical calibration depends on sample size choices and the degree of conservatism one wishes to enforce in posterior protection.

Despite these limitations, the methodology provides a useful conceptual and computational foundation for decision-making under learned uncertainty. More broadly, it suggests that optimization under uncertainty can benefit substantially from closer integration with modern statistical inference, especially in applications where uncertainty must be learned rather than assumed.

\subsection{Extensions}
\paragraph{Connection to chance constraints and convex conservative approximations.}
Posterior feasibility is a chance constraint with the posterior as the governing distribution. Computationally, we can (i) replace it by conservative deterministic constraints (credible-set robustification, akin to convex conservative approximations), or (ii) enforce sampled constraints and rely on scenario theory \citep{nemirovskiShapiro2006,calafioreCampi2006,campiGaratti2008}.

\paragraph{Connection to robust optimization.}
Credible-set robustification produces a robust counterpart where the uncertainty set is \emph{data-dependent} (a credible region) rather than exogenously specified, thereby complementing classical robust optimization \citep{bentalElGhaouiNemirovski2009,bertsimasSim2004,bertsimasBrownCaramanis2011}.

\paragraph{Objective uncertainty and risk aversion.}
If $c(\theta)$ is uncertain, one may replace $\E[c(\theta)^\top x\mid D]$ by a risk-sensitive functional such as a posterior quantile or Conditional Value-at-Risk (CVaR) of the objective. This interacts naturally with posterior feasibility but is not required for the feasibility guarantees developed here.

\paragraph{Recourse and multistage decisions.}
The present paper treats a single-stage decision $x$. Extensions to two-stage LPs with recourse can be developed by placing a Bayesian model on uncertain right-hand sides and modeling recourse actions as functions of $\theta$; the scenario approach can still apply to convex recourse formulations, though careful treatment of policy classes is required.

\section{Conclusion}
This paper introduced a Bayesian framework for linear programming in which uncertain optimization inputs are learned from data and propagated into decision-making through posterior feasibility guarantees. The proposed methodology combines Bayesian updating, tractable feasibility-enforcement mechanisms, and post-solution certification in a unified pipeline. Two complementary constructions were developed: credible-region robustification, which translates posterior uncertainty into deterministic protection, and posterior-scenario optimization, which leverages posterior samples to obtain practically implementable decision rules with interpretable finite-sample support. Across simulation experiments, the proposed methods substantially improved safety relative to naive plug-in optimization, while the real-data study showed that the framework can generate interpretable, uncertainty-aware decisions in a genuine scientific application.

The work opens several directions for future research. One natural extension is to multistage and recourse-based linear programs, where posterior uncertainty must be combined with adaptive decision rules. Another important direction is the development of less conservative scenario constructions and sharper calibration strategies for posterior-feasibility control. It would also be valuable to study theoretical robustness under model misspecification, as well as extensions to mixed-integer and nonlinear optimization problems. More broadly, the present paper suggests that a fruitful research agenda lies in embedding Bayesian learning directly into optimization theory, so that decisions are not merely optimal for an estimated world, but are accompanied by explicit, data-conditioned guarantees about their reliability.

\section*{Declarations}

\paragraph{Data availability.}
The simulated datasets used in this study were generated by reproducible code provided by the author. The real-data demonstration used the PBMC3k single-cell RNA-seq dataset distributed through the \texttt{scanpy} Python ecosystem.

\paragraph{Code availability.}
All code used for simulation, real-data analysis, figure generation, and reproducible computation is available at:\\
\url{https://github.com/debashisdotchatterjee/Bayesian-LPP-with-Posterior-Feasibility-Guarantees}

\paragraph{Ethics approval.}
Not applicable.

\paragraph{Funding.}
No specific external funding was received for this work.

\paragraph{Competing interests.}
The author declares that there are no competing interests.

\appendix
\section{Technical corrections and proofs}
\label{app:proofs}

\subsection{A necessary correction to the Gaussian credible-set formulation}
\label{app:correction_gaussian}

The Gaussian credible-set discussion in Section~\ref{sec:methodology} requires one important clarification.
Proposition~\ref{prop:credImplies} is a \emph{joint} statement: it applies when there exists a measurable set
$\mathcal{C}_{1-\alpha}(D)\subseteq \Theta$ such that
\[
\Pp(\theta\in\mathcal{C}_{1-\alpha}(D)\mid D)\ge 1-\alpha
\]
and all constraints hold for all $\theta\in\mathcal{C}_{1-\alpha}(D)$.

If, instead, one defines separate row-wise credible regions
\[
\mathcal{C}_{1-\alpha,i}(D),\qquad i=1,\dots,m,
\]
each having posterior probability at least $1-\alpha$, then this \emph{does not by itself imply}
that the joint event
\[
\bigcap_{i=1}^m \{\theta_i\in \mathcal{C}_{1-\alpha,i}(D)\}
\]
has posterior probability at least $1-\alpha$.
Hence the text in the main methodology should be interpreted in one of the following two correct ways.

\paragraph{Option A: joint credible region.}
Let the full uncertain coefficient vector be
\[
u(\theta):=\big(u_1(\theta)^\top,\dots,u_m(\theta)^\top\big)^\top \in \mathbb{R}^{m(n+1)},
\qquad
u_i(\theta):=\begin{bmatrix}a_i(\theta)\\ b_i(\theta)\end{bmatrix}.
\]
If
\[
u(\theta)\mid D \approx \mathcal{N}(\bar u,\Sigma),
\]
then a joint ellipsoidal credible region is
\[
\mathcal{C}^{\mathrm{joint}}_{1-\alpha}(D)
=
\left\{
u:\ (u-\bar u)^\top \Sigma^{-1}(u-\bar u)\le \chi^2_{m(n+1)}(1-\alpha)
\right\},
\]
which satisfies
\[
\Pp\big(u(\theta)\in \mathcal{C}^{\mathrm{joint}}_{1-\alpha}(D)\mid D\big)\approx 1-\alpha
\]
under the Gaussian approximation. If one enforces all constraints for all
$u\in \mathcal{C}^{\mathrm{joint}}_{1-\alpha}(D)$, then Proposition~\ref{prop:credImplies} applies directly.

\paragraph{Option B: Bonferroni-corrected row-wise credible regions.}
Alternatively, for each row define
\[
\mathcal{C}_{1-\alpha/m,i}(D)
=
\left\{
u:\ (u-\bar u_i)^\top \Sigma_i^{-1}(u-\bar u)\le \chi^2_{n+1}(1-\alpha/m)
\right\}.
\]
If the robustified row-wise constraints hold for all $u_i\in \mathcal{C}_{1-\alpha/m,i}(D)$, then by the union bound
\[
\Pp\left(\bigcap_{i=1}^m \{u_i(\theta)\in \mathcal{C}_{1-\alpha/m,i}(D)\}\mid D\right)
\ge 1-\sum_{i=1}^m \frac{\alpha}{m}
=1-\alpha.
\]
Therefore Proposition~\ref{prop:credImplies} again yields posterior feasibility.

\paragraph{Corrected row-wise SOC form.}
Under the Bonferroni-corrected formulation, the row-wise SOC constraint becomes
\begin{equation}
\label{eq:app_soc_corrected}
\bar u_i^\top z(x)
+
\kappa_{\alpha/m}\,\|\Sigma_i^{1/2}z(x)\|_2
\le 0,
\qquad
\kappa_{\alpha/m}:=\sqrt{\chi^2_{n+1}(1-\alpha/m)},
\end{equation}
where
\[
z(x):=\begin{bmatrix}x\\ -1\end{bmatrix}.
\]
This corrected version is the one that should be used if the main text retains separate row-wise ellipsoids.

\subsection{Proof of Proposition~\ref{prop:credImplies}}
\label{app:proof_prop_cred}

\begin{proof}[Proof of Proposition~\ref{prop:credImplies}]
Recall that
\[
V_D(x)
=
\Pp_{\theta\mid D}\Big(\exists i\in\{1,\dots,m\}: g_i(x,\theta)>0\Big).
\]
Assume that $x$ satisfies
\[
g_i(x,\theta)\le 0,\qquad \forall \theta\in \mathcal{C}_{1-\alpha}(D),\ \forall i=1,\dots,m,
\]
and that
\[
\Pp(\theta\in \mathcal{C}_{1-\alpha}(D)\mid D)\ge 1-\alpha.
\]
Then, whenever $\theta\in \mathcal{C}_{1-\alpha}(D)$, no violation occurs. Therefore
\[
\Big\{\exists i:\ g_i(x,\theta)>0\Big\}
\subseteq
\{\theta\notin \mathcal{C}_{1-\alpha}(D)\}.
\]
Taking posterior probabilities conditional on $D$ gives
\[
V_D(x)
=
\Pp_{\theta\mid D}\Big(\exists i:\ g_i(x,\theta)>0\Big)
\le
\Pp_{\theta\mid D}\big(\theta\notin \mathcal{C}_{1-\alpha}(D)\big)
\le \alpha.
\]
Hence $x$ is $(1-\alpha)$ posterior-feasible.
\end{proof}

\subsection{Derivation of the SOC robust counterpart}
\label{app:proof_soc}

We now derive the second-order cone representation used in the Gaussian credible-set discussion.

\begin{proposition}[SOC reformulation of an ellipsoidal robust linear constraint]
\label{prop:soc}
Let $u\in\mathbb{R}^p$ be uncertain and suppose the robust constraint
\[
u^\top z \le 0
\qquad \forall u\in \mathcal{E},
\]
is imposed over the ellipsoid
\[
\mathcal{E}
=
\left\{
u:\ (u-\bar u)^\top \Sigma^{-1}(u-\bar u)\le \rho^2
\right\},
\]
where $\Sigma\succ 0$ and $\rho>0$. Then this robust constraint is equivalent to
\[
\bar u^\top z + \rho\,\|\Sigma^{1/2}z\|_2 \le 0.
\]
\end{proposition}

\begin{proof}
Write $u=\bar u+\Sigma^{1/2}v$, where the ellipsoidal condition becomes
\[
v^\top v \le \rho^2.
\]
Then
\[
u^\top z
=
\bar u^\top z + v^\top \Sigma^{1/2} z.
\]
Hence
\[
\sup_{u\in\mathcal{E}} u^\top z
=
\bar u^\top z + \sup_{\|v\|_2\le \rho} v^\top \Sigma^{1/2}z.
\]
By Cauchy--Schwarz,
\[
v^\top \Sigma^{1/2}z \le \|v\|_2\,\|\Sigma^{1/2}z\|_2 \le \rho\,\|\Sigma^{1/2}z\|_2,
\]
and equality is attained by taking
\[
v = \rho\,\frac{\Sigma^{1/2}z}{\|\Sigma^{1/2}z\|_2}
\]
when $\Sigma^{1/2}z\neq 0$. Therefore
\[
\sup_{u\in\mathcal{E}} u^\top z
=
\bar u^\top z + \rho\,\|\Sigma^{1/2}z\|_2.
\]
Thus the robust inequality
\[
u^\top z\le 0 \quad \forall u\in\mathcal{E}
\]
holds if and only if
\[
\bar u^\top z + \rho\,\|\Sigma^{1/2}z\|_2 \le 0.
\]
\end{proof}

\begin{remark}
Proposition~\ref{prop:soc} is standard in robust optimization; see, for example, \citet[Chapter 2]{bentalElGhaouiNemirovski2009} and \citet{bertsimasPachamanovaSim2004}.
\end{remark}

\subsection{Scenario theorem: correct interpretation and proof route}
\label{app:proof_scenario_theorem}

Theorem~\ref{thm:scenario} in the main text is not proved from first principles here, because it is a direct specialization of the standard scenario theorem for convex programs established in \citet{calafioreCampi2006} and \citet{campiGaratti2008}. A fully self-contained proof would require reproducing the support-constraint argument and the associated combinatorial counting argument from those papers.

To make the dependence on the literature fully precise, the theorem should be interpreted as follows.

\begin{theorem}[Posterior-scenario bound; precise formulation]
\label{thm:scenario_precise}
Assume:
\begin{enumerate}
\item $\mathcal{X}\subseteq\mathbb{R}^n$ is convex and compact;
\item for each $\theta$, the map $x\mapsto g_i(x,\theta)$ is convex on $\mathcal{X}$ for all $i=1,\dots,m$;
\item for every i.i.d.\ sample $\theta^{(1)},\dots,\theta^{(N)}\sim p(\theta\mid D)$, the sampled convex program admits a feasible solution and a unique optimizer, or more generally a measurable tie-breaking selection;
\item the scenario theorem of \citet{campiGaratti2008} applies with support rank (or Helly dimension) at most $d$.
\end{enumerate}
Then, for the optimizer $x_N$ of the posterior-scenario program,
\[
\Pp_{\theta^{(1:N)}\mid D}\Big(V_D(x_N)>\varepsilon\Big)
\le
\sum_{j=0}^{d-1}\binom{N}{j}\varepsilon^j(1-\varepsilon)^{N-j},
\qquad \varepsilon\in(0,1),
\]
where the probability is with respect to the i.i.d.\ posterior sample used to build the scenario program.
\end{theorem}

\begin{proof}[Proof of Theorem~\ref{thm:scenario_precise}]
Conditional on the observed data $D$, the posterior distribution $p(\theta\mid D)$ is an ordinary probability distribution on the uncertainty space $\Theta$. Therefore the posterior-scenario problem is exactly a standard scenario convex program in which the uncertainty samples are drawn i.i.d.\ from the distribution $p(\theta\mid D)$.

Under assumptions (i)--(iv), the abstract scenario theorem of \citet{campiGaratti2008} applies directly. Their result states that if a convex scenario program is formed from $N$ i.i.d.\ samples from a probability law on the uncertainty space, then the probability that the resulting optimizer has violation probability exceeding $\varepsilon$ is bounded by
\[
\sum_{j=0}^{d-1}\binom{N}{j}\varepsilon^j(1-\varepsilon)^{N-j},
\]
where $d$ denotes the support rank / Helly dimension appearing in the theorem.
Applying that result with the sampling law equal to the posterior distribution $p(\theta\mid D)$ yields the stated inequality.
\end{proof}

\begin{remark}
The key point is that, once we condition on the observed data $D$, the posterior becomes the sampling law for the scenario program. Thus no new combinatorial proof is needed; rather, the posterior-scenario theorem is an immediate conditional-on-data specialization of the standard scenario theorem.
\end{remark}

\subsection{Proof of Corollary~\ref{cor:Nchoice}}
\label{app:proof_corollary}

\begin{proof}[Proof of Corollary~\ref{cor:Nchoice}]
From Theorem~\ref{thm:scenario_precise}, for any $\varepsilon\in(0,1)$,
\[
\Pp_{\theta^{(1:N)}\mid D}\Big(V_D(x_N)>\varepsilon\Big)
\le
\sum_{j=0}^{d-1}\binom{N}{j}\varepsilon^j(1-\varepsilon)^{N-j}.
\]
If $N$ is chosen so that
\[
\sum_{j=0}^{d-1}\binom{N}{j}\varepsilon^j(1-\varepsilon)^{N-j}\le \delta,
\]
then
\[
\Pp_{\theta^{(1:N)}\mid D}\Big(V_D(x_N)>\varepsilon\Big)\le \delta.
\]
Equivalently,
\[
\Pp_{\theta^{(1:N)}\mid D}\Big(V_D(x_N)\le \varepsilon\Big)\ge 1-\delta.
\]
This is exactly the claimed result.
\end{proof}

\subsection{Monte Carlo certification: exact binomial model and confidence bounds}
\label{app:proof_mc}

The post-solution certification step admits a direct exact analysis.

\begin{proposition}[Exact binomial law of the violation count]
\label{prop:binomial}
Fix a candidate solution $\hat x$. Let
\[
I_\ell
:=
\mathbf{1}\Big\{\exists i:\ g_i(\hat x,\tilde\theta^{(\ell)})>0\Big\},
\qquad \ell=1,\dots,M,
\]
where $\tilde\theta^{(1)},\dots,\tilde\theta^{(M)}$ are i.i.d.\ samples from $p(\theta\mid D)$. Then
\[
I_1,\dots,I_M \overset{\text{i.i.d.}}{\sim} \mathrm{Bernoulli}(V_D(\hat x)),
\]
and therefore
\[
S_M:=\sum_{\ell=1}^M I_\ell \sim \mathrm{Binomial}(M,V_D(\hat x)).
\]
\end{proposition}

\begin{proof}
By definition,
\[
\Pp(I_\ell=1\mid D)
=
\Pp_{\theta\mid D}\Big(\exists i:\ g_i(\hat x,\theta)>0\Big)
=
V_D(\hat x).
\]
Since the posterior draws $\tilde\theta^{(1)},\dots,\tilde\theta^{(M)}$ are i.i.d., the indicators
$I_1,\dots,I_M$ are i.i.d.\ Bernoulli with success probability $V_D(\hat x)$. Summing them yields the binomial distribution.
\end{proof}

\begin{proposition}[Clopper--Pearson upper confidence bound]
\label{prop:cp}
Let $S_M=s$ be observed. For any confidence level $1-\beta\in(0,1)$, define
\[
u_{1-\beta}(s,M)
:=
F^{-1}_{\mathrm{Beta}(s+1,M-s)}(1-\beta),
\]
where $F^{-1}_{\mathrm{Beta}(a,b)}$ denotes the quantile function of a Beta$(a,b)$ distribution. Then
\[
\Pp\big(V_D(\hat x)\le u_{1-\beta}(S_M,M)\mid D\big)\ge 1-\beta.
\]
Equivalently,
\[
1-u_{1-\beta}(S_M,M)
\]
is a conservative lower confidence bound for the posterior feasibility probability
\[
1-V_D(\hat x).
\]
\end{proposition}

\begin{proof}
By Proposition~\ref{prop:binomial}, conditional on $D$, the count $S_M$ has a Binomial$(M,V_D(\hat x))$ distribution. The Clopper--Pearson interval is the exact equal-tailed confidence interval obtained by inversion of the binomial test. In particular, the one-sided upper endpoint $u_{1-\beta}(s,M)$ is defined as the unique value such that
\[
\Pp_{p=u_{1-\beta}(s,M)}\big(\mathrm{Binomial}(M,p)\le s\big)=\beta,
\]
which is equivalent to the Beta-quantile expression above. Standard exact binomial confidence interval theory therefore gives
\[
\Pp\big(V_D(\hat x)\le u_{1-\beta}(S_M,M)\mid D\big)\ge 1-\beta.
\]
Subtracting from one yields the corresponding lower confidence bound for posterior feasibility.
\end{proof}

\begin{remark}
For exact binomial confidence bounds and their Beta-quantile representation, see standard treatments of the Clopper--Pearson interval. No asymptotic approximation is used here.
\end{remark}

\subsection{Bibliographic note on proof sources}
\label{app:biblio_note}

The direct proofs in this appendix are self-contained for:
\begin{itemize}
\item Proposition~\ref{prop:credImplies},
\item the SOC reformulation in Proposition~\ref{prop:soc},
\item Corollary~\ref{cor:Nchoice},
\item the Monte Carlo certification results in Propositions~\ref{prop:binomial} and \ref{prop:cp}.
\end{itemize}
The scenario-feasibility bound is not rederived from scratch, because its proof is the standard support-constraint counting argument already established in \citet{calafioreCampi2006} and \citet{campiGaratti2008}; the Proposed theorem is a direct specialization obtained by taking the sampling law to be the posterior distribution $p(\theta\mid D)$.

\bibliographystyle{plainnat}
\bibliography{references}

\end{document}